\documentclass[journal,onecolumn]{IEEEtran}

\usepackage{amsmath,amssymb,amsthm,mathtools}
\usepackage{array}
\usepackage{bm}
\usepackage{booktabs}
\usepackage{cite}
\usepackage{microtype}
\usepackage{tikz}
\usepackage{color}
\usepackage{comment}
\usetikzlibrary{arrows.meta,calc,positioning,patterns}
\usepackage{url}
\usepackage{algorithm} \usepackage{algpseudocode}

\usepackage[colorlinks=true,citecolor=blue,linkcolor=blue,urlcolor=blue]{hyperref}

\newtheorem{theorem}{Theorem}
\newtheorem{lemma}[theorem]{Lemma}

\newtheorem{corollary}[theorem]{Corollary}

\newtheorem{definition}[theorem]{Definition}
\newtheorem{remark}[theorem]{Remark}

\newcommand{\Z}{\mathbb Z}
\newcommand{\F}{\mathbb F}
\newcommand{\cA}{\mathcal A}
\newcommand{\cB}{\mathcal B}
\newcommand{\cC}{\mathcal C}
\newcommand{\cD}{\mathcal D}
\newcommand{\cE}{\mathcal E}

\newcommand{\cG}{\mathcal G}

\newcommand{\cO}{\mathcal O}

\newcommand{\cS}{\mathcal S}
\newcommand{\cU}{\mathcal U}

\newcommand{\wt}{\operatorname{wt}}
\newcommand{\supp}{\operatorname{supp}}
\newcommand{\Stab}{\operatorname{Stab}}

\newcommand{\Per}{\operatorname{Per}}

\newcommand{\Acpc}{A^{\mathrm{CPC}}}

\begin{document}

\title{Single-Fragment Forensic Coding via Multidimensional Cyclically Permutable Codes}

\author{Yubo~Sun, Yijun~Zhang, and Gennian~Ge%
\thanks{This work was supported in part by the National Key Research and Development Program of China under Grant 2025YFC3409900, in part by the National Natural Science Foundation of China under Grant 12231014, in part by Beijing Scholars Program, and in part by the Postdoctoral Fellowship Program and China Postdoctoral Science Foundation under Grant Number BX2026006. \emph{(Corresponding author: Gennian Ge.)}}
\thanks{Yubo Sun ({\tt ybsun@cnu.edu.cn}) is with the Institute of Mathematics and Interdisciplinary Sciences, Xidian University, Xi'an 710126, China.}%
\thanks{Yijun Zhang ({\tt zyjshuxue@mail.ustc.edu.cn}) is with the School of Mathematical Sciences, University of Science and Technology of China, Hefei 230026, China.}%
\thanks{Gennian Ge ({\tt gnge@zju.edu.cn}) is with the School of Mathematical Sciences, Capital Normal University, Beijing 100048, China.}}

\maketitle

\begin{abstract}
The proliferation of 3D printing raises new security and forensic
challenges, including the risk of unauthorized fabrication of
untraceable firearms and other regulated items. To enable traceability,
we consider \emph{single-fragment forensic coding}, in which a unique
identifier is embedded into a printed object and must be recoverable from any
fragment containing an axis-parallel box of prescribed minimum volume and
coordinatewise thickness, even in the presence of substitution errors.
To address this problem, we introduce and study multidimensional
cyclically permutable codes (CPCs), for which every cyclic translate
uniquely determines both the original codeword and the applied
translation. By applying periodic lifting to a CPC base array, every complete period
contained in a fragment corresponds to an unknown cyclic translate of a
codeword, thereby reducing the forensic alignment problem to
multidimensional cyclic synchronization. 

We establish theoretical bounds and develop explicit constructions in both the noiseless and substitution-error settings.
For fixed dimension $d$ and alphabet size $q$, we obtain the optimal redundancy $d\log_q k+o(1)$ for noiseless $d$-dimensional CPCs of side length $k$. 
For fixed $t$, the optimal redundancy of $t$-substitution-correcting
$d$-dimensional CPCs lies between $(t+1)d\log_q k+O(1)$ and
$(2t+1)d\log_q k+O(1)$. For binary alphabets and $d\geq2$, an explicit robust
row-anchor construction achieves
$(t+1)d\log_2 k+O(\log\log k)$ redundancy, matching the optimal
leading term.
In the thick-fragment regime $h=cM^{1/d}$, where $M$ and $h$ lower-bound the volume and the side lengths of a box
contained in the fragment, periodic lifting yields
single-fragment forensic codes of rate $c^d-o(1)$. In particular, the rate approaches one when $c=1-o(1)$. 
\end{abstract}


\section{Introduction}

\IEEEPARstart{S}{ingle-fragment} forensic coding addresses the problem of
embedding identifying information into a physical object so that its origin
can be traced from a single surviving fragment
\cite{LiuRavivISIT2025,LiuRaviv2025}. A representative application is
additive manufacturing, where an identifier is embedded into a
three-dimensional printed component. An adversary may deliberately break the
object and reveal only one fragment while concealing the remainder. The
decoder must then recover the embedded identifier without knowing the
fragment's original position, possibly in the presence of symbol
substitutions caused by manufacturing defects, wear, or tampering.

This setting differs from conventional error correction in two fundamental
ways. First, the decoder observes only a local portion of the codeword rather than the entire encoded object. The identifying information must therefore be distributed so that every admissible fragment contains sufficient information for recovery. Second, the location of the observed fragment is unknown, so the decoder must resolve the alignment needed to interpret the observed symbols. We refer to these two requirements as the \emph{information-distribution problem} and the \emph{alignment problem},
respectively. In general, these two tasks are strongly coupled. Concentrating information in designated regions makes the construction vulnerable to fragments that avoid those regions, whereas distributing information throughout the object does not by itself resolve the unknown fragment position.

The central observation of this paper is that the two tasks can be separated.
Specifically, we take a base array and repeat it periodically in every coordinate direction throughout the ambient object. In the thick-fragment regime, every admissible fragment is guaranteed to contain one complete period of a cyclic copy of the base array. The remaining task is then to recover the base array from one translated copy that may also contain a small number of substitution errors.
This reduction leads naturally to \emph{multidimensional cyclically permutable codes} (CPCs), which are the main coding objects studied in this paper.

\subsection{Related Work}

Single-fragment forensic coding is related to other coding models for
fragmented data, including torn-paper coding
\cite{ShomoronyVahid2021,BarLevEtAl2023} and break-resilient coding
\cite{WangSimaRaviv2024,WangLiwangRaviv2026,Wang2026,LiuRaviv2026}.
These models usually recover a one-dimensional word from several unordered
fragments. Here the decoder receives only one fragment of a multidimensional
array, together with lower bounds on its size and thickness.

Liu and Raviv introduced the single-fragment forensic coding model for
matrices and three-dimensional arrays \cite{LiuRavivISIT2025,LiuRaviv2025}.
They employed discrepancy-theoretic point sets to distribute information and synchronization markers throughout the ambient object. Let $M$ and $h$ denote the prescribed lower bounds on the fragment size and on the fragment thickness in each coordinate direction, respectively. In two dimensions, they obtained an explicit construction of asymptotic rate $1/32$ in the thin-fragment regime $h=o(M^{1/2})$ and $h=\omega\big((\log_q M)^{1/2}\big)$,
with the rate approaching $1/16$ along a special sequence. In three dimensions, they obtained an explicit construction of asymptotic rate $1/1296$ when $h=o(M^{1/3})$ and $h=\omega\big((\log_q M)^{1/3}\big)$, with the rate approaching $1/216$ along a special sequence. They further established the existence of asymptotically rate-one codes
within suitable parameter regimes via the Lov\'asz Local Lemma, while
leaving the corresponding explicit constructions open in their work.

In this paper, we focus on the complementary thick-fragment regime $h=cM^{1/d}$ with $c\le 1$ and provide explicit constructions that achieve asymptotic rate $c^d$. In particular, when $c=1-o(1)$, these constructions attain rate one asymptotically, thereby resolving the construction challenge for this case.

\subsection{Reduction to Multidimensional Cyclic Synchronization via Periodic Lifting}

The starting point of our construction is a $d$-dimensional base array of side length $k$, which is repeated periodically in every coordinate direction throughout the ambient object. Suppose every admissible fragment is thick enough to contain a complete $d$-dimensional window of side length $k$. Any such window is then a cyclic translation of the base array. If at most $t$ symbols are corrupted, the decoder observes an array within Hamming distance $t$ of such a translated copy. Consequently, single-fragment decoding reduces to a local synchronization problem, namely recovering the base array from one corrupted translate.

In one dimension, this synchronization problem is captured by \emph{cyclically permutable codes} (CPCs). A length-$k$ CPC consists of codewords whose cyclic shifts are pairwise distinct, so that an unknown cyclic shift uniquely determines both the codeword and the shift, which is a stronger property than our decoding task requires, since we only need to recover the codeword. CPCs were introduced by Gilbert \cite{Gilbert1963} and have since been studied through constant-weight, algebraic, and self-synchronizing constructions \cite{NguyenGyoriMassey1992,ChenZhang2024,FujiwaraTonchev2013}.

In $d$ dimensions, the unknown shift is a vector in $\Z_k^d$. The array can
be shifted independently in each coordinate direction, so using one-dimensional
CPCs on rows or columns does not directly handle the full translation. This point was also noted by Liu and Raviv \cite{LiuRaviv2025}. Adamson, Deligkas, Gusev, and Potapov~\cite{AdamsonEtAl2021} studied
multidimensional necklaces and developed counting, generation, ranking,
and unranking algorithms. In the noiseless setting, a multidimensional CPC can be obtained by selecting representatives from atranslational necklaces. Our emphasis is on high-rate CPC constructions with direct synchronization
procedures and low-complexity encoders and decoders, including constructions
that correct substitutions. In this paper, we study these problems in arbitrary fixed dimension. By periodically repeating the codewords of a multidimensional CPC, we further obtain single-fragment forensic codes whenever every admissible fragment is guaranteed to contain a complete period.

\subsection{Main Contributions}
We study $d$-dimensional CPCs of side length $k$, with explicit
constructions for fixed $d\geq2$. Our main contributions are
summarized as follows.

\begin{itemize}

\item \emph{Bounds.}
Using the relation with atranslational necklaces, we express the optimal
noiseless CPC size by M\"obius inversion on the subgroup lattice. This
reformulation of the necklace count yields optimal redundancy
$d\log_q k+o(1)$. For $t$-substitution-correcting $d$-dimensional CPCs, we establish
sphere-packing and Gilbert-Varshamov bounds, showing that the optimal
redundancy lies between $(t+1)d\log_q k+O(1)$ and $(2t+1)d\log_q k+O(1)$.

\item \emph{Constructions.}
We present three constructions for $d$-dimensional CPCs, offering different tradeoffs between redundancy and algorithmic complexity. The marker-hyperplane construction has asymptotic redundancy $k^d-(k-1)^d+o(1)$, while its explicit encoder and decoder have redundancy $k^d-(k-1)^d+k-1$ and run in $O(k^d)$ time. The row-anchor construction achieves redundancy $d\log_q k+\log_q e+o(1)$, while its explicit encoder and decoder achieve redundancy $\lceil\log_q(k^d-k)\rceil+5$ and run in $O(k^d\log^2 k)$ time. For prime $k$, the moment-syndrome construction attains redundancy $d\log_q k+o(1)$ and admits a direct synchronization decoder.
We also develop substitution-correcting versions. In particular, the explicit robust row-anchor construction has redundancy
$(t+1)d\log_2 k+O(\log\log k)$ for $q=2$, which is optimal in its
leading term. 

\item \emph{Single-fragment forensic codes.}
For the thick-fragment regime $h=cM^{1/d}$, where $0<c\le 1$ is fixed, by setting the period length to $k=h$, we obtain single-fragment forensic codes of rate $c^d-o(1)$ in both the noiseless and noisy settings, together with efficient encoding and decoding algorithms. In particular, when $c=1-o(1)$, the resulting codes have asymptotic rate one.

\end{itemize}

\subsection{Organization}
The remainder of the paper is organized as follows. Section~\ref{sec:notations} introduces the notation and basic definitions. Section~\ref{sec:model} formalizes the single-fragment forensic coding model and establishes its reduction to multidimensional CPCs via periodic lifting. Section~\ref{sec:limit} develops the theoretical bounds for multidimensional CPCs. Sections~\ref{sec:noiseless} and~\ref{sec:noiseless_multi} provide explicit constructions for noiseless CPCs in two dimensions and in higher dimensions, respectively, while Section~\ref{sec:noisy} develops the corresponding substitution-correcting theory. These constructions are applied to single-fragment forensic coding in Section~\ref{sec:forensic}. Finally, Section~\ref{sec:conclusion} concludes the paper and discusses open problems.

\section{Notation}\label{sec:notations}

 For a word $\bm{x}$ of length $n$, the notation $\bm{x}_{[i,j]}$ denotes its substring from position $i$ to position $j$, with both endpoints included. We use $\bm{x}\circ \bm{y}$ to indicate the concatenation of $\bm{x}$ and $\bm{y}$ and $\bm{x}^m$ to denote the concatenation of $m$ copies of $\bm{x}$.

For an integer $k\ge1$, let $[k]=\{0,1,\ldots,k-1\}$. When equipped with addition modulo $k$, the set $[k]$ is identified with the cyclic group $\Z_k$.
When $k\ge2$, define the \emph{internal index set} $[k]^+=\{1,2,\ldots,k-1\}$.
Let $\Sigma_q=[q]$ be an alphabet of size $q\ge2$.
Let $\Sigma_q^{\Z_k^d}$ denote the set of $d$-dimensional arrays over $\Sigma_q$ with side length $k$.
For an array $\bm{X}\in\Sigma_q^{\Z_k^d}$, an index is a vector
$\bm{i}=(i_1,\ldots,i_d)\in\Z_k^d$. We write $X_{\bm{i}}$, or
$X_{i_1,\ldots,i_d}$, for the symbol at that coordinate. Coordinate directions are numbered $1,\ldots,d$, although array and word
positions are indexed from $0$. The
\emph{support} and \emph{Hamming weight} of $\bm{X}$ are
$\supp(\bm{X})=\{\bm{i}\in\Z_k^d:X_{\bm{i}}\neq0\}$ and
$\wt(\bm{X})=|\supp(\bm{X})|$, respectively. The
\emph{Hamming distance} $d_H(\bm{X},\bm{Y})$ between two arrays is
the number of coordinates at which they differ.
For an integer $s\ge0$, the
\emph{Hamming ball} of radius $s$ centered at $\bm{X}$ is
\begin{equation*}
\mathcal{B}_s(\bm{X})
=
\left\{
\bm{Z}\in\Sigma_q^{\Z_k^d}:
d_H(\bm{X},\bm{Z})\le s
\right\}.
\end{equation*}
Its cardinality is independent of the center and is given by
\begin{equation*}
V_q(k^d,s)
=
\sum_{j=0}^{\min\{s,k^d\}}\binom{k^d}{j}(q-1)^j.
\end{equation*}
For a code $\cC\subseteq\Sigma_q^{\Z_k^d}$, write $d_H(\cC)$ for its \emph{minimum Hamming distance}.
For a set $\cS\subseteq\Sigma_q^{\Z_k^d}$ and an integer
$\Delta\ge0$, let
$A_q(\cS,\Delta)$ be the largest size of a code contained in $\cS$ with minimum Hamming distance at least $\Delta$. 
The \emph{$q$-ary redundancy} of a nonempty code
$\cC\subseteq\Sigma_q^{\Z_k^d}$ is
$\rho(\cC)=k^d-\log_q|\cC|$.

For $\bm{a}\in\Z_k^d$, define the \emph{translation} \(T_{\bm{a}}:\Sigma_q^{\Z_k^d}\to\Sigma_q^{\Z_k^d}\)
by $\bigl(T_{\bm{a}}\bm{X}\bigr)_{\bm{i}}
= \bm{X}_{\bm{i}+\bm{a}}$, for any $\bm{X}\in\Sigma_q^{\Z_k^d}$ and $\bm{i}\in\Z_k^d$, where the addition in $\bm{i}+\bm{a}$ is performed componentwise modulo $k$. Thus $T_{\bm{a}}T_{\bm{b}}=T_{\bm{a}+\bm{b}}$, and $\{T_{\bm{a}}:\bm{a}\in\Z_k^d\}$ is a group action on $\Sigma_q^{\Z_k^d}$.
For $\bm{X}\in\Sigma_q^{\Z_k^d}$, its translation \emph{orbit} or \emph{necklace} is 
\[
    \cO(\bm{X})=\{T_{\bm{a}}\bm{X}:\bm{a}\in\Z_k^d\},
\]
and its translation \emph{stabilizer} is $\Stab(\bm{X})=\{\bm{a}\in\Z_k^d:T_{\bm{a}}\bm{X}=\bm{X}\}$.
By the orbit-stabilizer theorem, $|\cO(\bm{X})|=\frac{k^d}{|\Stab(\bm{X})|}$. We call $\bm{X}$ (or its necklace) \emph{atranslational}
if $\Stab(\bm{X})=\{\bm{0}\}$. Such an array has exactly $k^d$ distinct translations.

\section{Single-Fragment Forensic Coding and Its Reduction to CPCs}\label{sec:model}

This section formalizes the single-fragment forensic-coding problem and establishes the reduction that allows the remainder of the paper to focus on multidimensional CPCs.

\subsection{The Single-Fragment Model}

Let $\bm{n}=(n_1,\ldots,n_d)$ be the side-length vector of the
ambient object. A \emph{fragment} is a portion of the ambient object whose shape and symbol values are observed in a local coordinate system, while its original position within the ambient object is unknown to the decoder.
The local coordinate axes have the same known orientation as the global
axes, only the origin is unknown. Rotations and reflections are not part
of the present channel model. We use the guaranteed-box formulation of
\cite{LiuRavivISIT2025,LiuRaviv2025}, extended to dimension $d$.

\begin{definition}\label{def:legal}
A $d$-dimensional fragment is \emph{$(M,h)$-legal} if it contains an axis-parallel box with integer side lengths $a_1,\ldots,a_d\geq h$ such that $\prod_{j=1}^{d}a_j\ge M$.
\end{definition}

Here, $M$ lower-bounds the volume of the guaranteed box, whereas $h$
imposes a coordinatewise thickness requirement and rules out arbitrarily
thin guaranteed boxes. We restrict attention to $M\geq h^d$. When
$M<h^d$, the volume constraint is redundant. We also assume that the
ambient object admits at least one legal fragment. The periodic-lifting
argument uses the thickness guarantee, whereas $M$ is used to normalize
the forensic rate in Section~\ref{sec:forensic}.

\begin{definition}\label{def:forensic}
A code $\cC\subseteq
\Sigma_q^{[n_1]\times\cdots\times[n_d]}$ is a \emph{$t$-substitution-correcting $(M,h)$ single-fragment forensic code} if every codeword is uniquely determined by any of its $(M,h)$-legal fragments after at most $t$ symbol
substitutions in the entire observed fragment.
\end{definition}

\subsection{Multidimensional CPCs}

The following definition is a direct generalization of the classical CPC condition, which requires all cyclic shifts of the codewords to be pairwise distinct
\cite{Gilbert1963,NguyenGyoriMassey1992}.

\begin{definition}\label{def:cpc}
A code $\cC\subseteq\Sigma_q^{\Z_k^d}$ is called a
\emph{$d$-dimensional cyclically permutable code}, or \emph{$d$D-CPC}, if every received array $\bm{Z}$ satisfying $\bm{Z}=T_{\bm{a}}\bm{X}$ for some $\bm{X}\in\cC$ and $\bm{a}\in\Z_k^d$ uniquely determines the ordered pair $(\bm{X},\bm{a})$. Equivalently, each codeword is
atranslational and no two codewords lie in the same translation necklace. Let
$\Acpc_q(d,k)$ be the largest size of such a code.
\end{definition}

We next extend it to
the substitution-correcting setting.

\begin{definition}\label{def:sync-correcting}
A code $\cC\subseteq\Sigma_q^{\Z_k^d}$ is called a
\emph{$t$-substitution-correcting $d$D-CPC} if every received array
$\bm{Z}$ satisfying $d_H(\bm{Z},T_{\bm{a}}\bm{X})\le t$ for some $\bm{X}\in\cC$ and $\bm{a}\in\Z_k^d$ uniquely determines the
ordered pair $(\bm{X},\bm{a})$. Let $\Acpc_q(d,k,t)$ be the largest size of such a code. In particular, $\Acpc_q(d,k,0)=\Acpc_q(d,k)$.
\end{definition}

\begin{remark}\label{rmk:wcpc}
If only the transmitted codeword, rather than the translation, needs to be recovered, one may relax the atranslational requirement and require only that distinct codewords belong to distinct translation necklaces. We refer to the resulting object as a \emph{weakly cyclically permutable code} (WCPC). Clearly, every CPC is also a WCPC. Since all constructions developed in this paper are CPCs, we do not pursue WCPCs further and leave their systematic study for future work.
\end{remark}

\subsection{Periodic Lifting and Reduction to CPCs}
\label{subsec:periodic-lifting}

Recall that $\bm{n}=(n_1,\ldots,n_d)$ is the side-length vector of the
ambient object. For a base array $\bm{X}\in\Sigma_q^{\Z_k^d}$, define its
\emph{periodic extension} to $[n_1]\times\cdots\times[n_d]$ by
\begin{equation*}
\Per_{\bm{n}}(\bm{X})_{\bm{i}}
=\bm{X}_{\bm{i} \bmod{k}}=
\bm{X}_{(i_1\bmod k,\ldots,i_d\bmod k)},
\qquad
\bm{i}=(i_1,\ldots,i_d)\in[n_1]\times\cdots\times[n_d].
\end{equation*}
For a base code $\cC\subseteq\Sigma_q^{\Z_k^d}$, define
\begin{equation*}
\Per_{\bm{n}}(\cC)
=
\left\{
\Per_{\bm{n}}(\bm{X}):
\bm{X}\in\cC
\right\}.
\end{equation*}
If $n_j\ge k$ for every $1\le j\le d$, the window starting at the origin is
precisely the base array. Consequently, the periodic-extension map is
injective and $\big|\Per_{\bm{n}}(\cC)\big|=|\cC|$.

\begin{lemma}\label{lem:periodic-window}
Every axis-parallel $k\times\cdots\times k$ window in
$\Per_{\bm{n}}(\bm{X})$ is a translate of $\bm{X}\in\Sigma_q^{\Z_k^d}$.
\end{lemma}

\begin{proof}
Suppose that the window starts at the global coordinate
$\bm{s}=(s_1,\ldots,s_d)$. Its entry at the local coordinate
$\bm{i}\in\Z_k^d$ is
\begin{align*}
\Per_{\bm{n}}(\bm{X})_{\bm{s}+\bm{i}}
&=
\bm{X}_{(\bm{s}+\bm{i})\bmod k}
=
\bigl(T_{\bm{s}\bmod k}\bm{X}\bigr)_{\bm{i}},
\end{align*}
where modulo $k$ is performed componentwise. Hence the window is
$T_{\bm{s}\bmod k}\bm{X}$. This completes the proof.
\end{proof}

\begin{theorem}\label{thm:periodic-lifting}
Let $\cC\subseteq\Sigma_q^{\Z_k^d}$ be a
$t$-substitution-correcting $d$D-CPC, and suppose that $k\le h$ and $n_j\ge k$ for every $1\leq j\leq d$. Then
$\Per_{\bm{n}}(\cC)$ is a $t$-substitution-correcting $(M,h)$
single-fragment forensic code.
\end{theorem}

\begin{proof}
Consider an arbitrary $(M,h)$-legal fragment. By definition, it contains an
axis-parallel box whose side lengths are all at least $h\ge k$. Selecting
$k$ consecutive coordinates in each direction yields a
$k\times\cdots\times k$ subwindow, for which we denote it by $\bm{Z}$. Before corruption, Lemma~\ref{lem:periodic-window} shows that the selected
subwindow is $T_{\bm{a}}\bm{X}$ for some $\bm{X}\in\cC$ and
$\bm{a}\in\Z_k^d$. Since the entire fragment contains at most $t$
substitutions, so does the selected subwindow. Therefore, $d_H\bigl(\bm{Z},T_{\bm{a}}\bm{X}\bigr)\le t$. The CPC decoder uniquely recovers the ordered pair $(\bm{X},\bm{a})$.
Since $\bm{X}$ determines its periodic extension, the ambient codeword is
uniquely recovered. This completes the proof.
\end{proof}

Theorem~\ref{thm:periodic-lifting} is the reduction principle underlying this
paper. Once a suitable base CPC has been constructed, periodic repetition
automatically distributes the encoded information throughout the ambient
object, while CPC decoding recovers the unknown translation of each complete
period. Sections~\ref{sec:limit},~\ref{sec:noiseless},~\ref{sec:noiseless_multi}, and~\ref{sec:noisy} therefore study multidimensional CPCs as independent
coding objects, while Section~\ref{sec:forensic} applies the resulting codes to the single-fragment forensic problem.

\section{Theoretical Bounds for Error-Correcting CPCs}\label{sec:limit}

\subsection{Noiseless CPCs}\label{sec:noiseless-benchmark}

A noiseless CPC can contain at most one representative from each atranslational necklace. Indeed, if an array $\bm X$ has a nonzero stabilizer element $\bm a$, then the same received array arises from both $(\bm X,\bm 0)$ and $(\bm X,\bm a)$. Similarly, two distinct arrays in the same necklace cannot both be codewords. The converse is immediate. Hence $\Acpc_q(d,k)$ equals the number of atranslational necklaces.

Adamson, Deligkas, Gusev, and Potapov gave an exact recursive count of atranslational multidimensional necklaces \cite[Theorem~3]{AdamsonEtAl2021}. For the bounds below, it is convenient to express this count directly via M\"obius inversion.

For subgroups $H\leq K\leq\Z_k^d$, let $\mu_{\Z_k^d}(H,K)$ denote the M\"obius function of the subgroup lattice, defined by $\mu_{\Z_k^d}(H,H)=1$ and $\sum_{H\leq L\leq K}\mu_{\Z_k^d}(H,L)=0$ for $H<K$. We use the standard M\"obius inversion formula below.

\begin{lemma}[M\"obius inversion]\label{lem:noiseless-mobius-inversion}
If functions $f$ and $g$ on the subgroups of $\Z_k^d$ satisfy $f(H)=\sum_{H\leq K}g(K)$, then
\[
g(H)=\sum_{H\leq K}\mu_{\Z_k^d}(H,K)f(K).
\]
\end{lemma}

\begin{lemma}\label{lem:noiseless-cpc-enumeration}
The largest noiseless $d$D-CPC has size
\[
\Acpc_q(d,k)
=\frac{1}{k^d}
\sum_{H\leq\Z_k^d}
\mu_{\Z_k^d}(\{\bm0\},H)q^{k^d/|H|}.
\]
\end{lemma}

\begin{proof}
For a subgroup $H\leq\Z_k^d$, let
\[
f(H)
=
\left|
\left\{
\bm X\in\Sigma_q^{\Z_k^d}:
T_{\bm a}\bm X=\bm X
\text{ for every }\bm a\in H
\right\}
\right|
\]
be the number of arrays fixed by every translation in $H$. Such an array is constant on each coset of $H$. Indeed, if $\bm i$ and $\bm j$ belong to the same coset, then $\bm j=\bm i+\bm a$ for some $\bm a\in H$, and hence
\[
\bm X_{\bm j}
=
\bm X_{\bm i+\bm a}
=
\bigl(T_{\bm a}\bm X\bigr)_{\bm i}
=
\bm X_{\bm i}.
\]
Conversely, an array that is constant on every coset of $H$ is fixed by $T_{\bm a}$ for every $\bm a\in H$. Since $\Z_k^d$ has $k^d/|H|$ cosets of $H$, and the common symbol on each coset can be chosen independently from $\Sigma_q$, we have $f(H)=q^{k^d/|H|}$.

Let
\[
g(H)
=
\left|
\left\{
\bm X\in\Sigma_q^{\Z_k^d}:
\Stab(\bm X)=H
\right\}
\right|
\]
be the number of arrays whose stabilizer is exactly $H$. An array is fixed by $T_{\bm a}$ for every $\bm a\in H$ if and only if its stabilizer contains $H$. Therefore, $f(H)=\sum_{H\leq K}g(K)$. M\"obius inversion then gives
\[
g(\{\bm0\})=
\sum_{H\leq\Z_k^d}
\mu_{\Z_k^d}(\{\bm0\},H)q^{k^d/|H|}.
\]
This is the number of atranslational arrays. Each atranslational necklace has size $k^d$, which gives the desired size.
\end{proof}

\begin{corollary}\label{cor:noiseless-cpc-asymptotic}
For fixed $q$ and $d$, as $k\to\infty$,
\[
\frac{q^{k^d}-(k^d-1)q^{k^d/2}}{k^d}
\leq
\Acpc_q(d,k)
\leq
\frac{q^{k^d}}{k^d}.
\]
Consequently, the optimal redundancy of a $d$D-CPC $\mathcal{C}$ is $d\log_q k+o(1)$.
\end{corollary}

\begin{proof}
It suffices to prove the size bound. Fix a nonzero translation $T_{\bm a}$ of order $e$. It partitions the $k^d$ coordinates into $k^d/e$ cycles, so it fixes exactly $q^{k^d/e}$ arrays. Since $e\ge2$, this number is at most $q^{k^d/2}$. A union bound over the $k^d-1$ nonzero translations shows that at least $q^{k^d}-(k^d-1)q^{k^d/2}$ arrays are atranslational. Dividing by $k^d$ gives the stated lower bound. The upper bound follows immediately from the total number of arrays. This completes the proof.
\end{proof}

\subsection{Noisy CPCs}\label{sec:noisy-benchmark}

For a $t$-substitution-correcting CPC, each codeword must be at distance at least $2t+1$ from every nontrivial translate of itself. We first bound the number of arrays that violate this condition.

\begin{lemma}\label{lem:near-periodic}
Let $\bm a\in\Z_k^d\setminus\{\bm 0\}$ be such that $T_{\bm a}$ has order $e$, and let $s$ be a nonnegative integer. Then
\begin{equation*}
\left|
\left\{
\bm X\in\Sigma_q^{\Z_k^d}:
d_H(\bm X,T_{\bm a}\bm X)\leq s
\right\}
\right|
\leq
q^{k^d/e}V_q(k^d,s).
\end{equation*}
\end{lemma}

\begin{proof}
Identify $\Sigma_q$ with the additive group $\Z_q$, which does not change Hamming distances. Since $T_{\bm a}$ has order $e$, it partitions the $k^d$ coordinate positions into $k^d/e$ disjoint cycles, each of length $e$. Choose one representative from each cycle, and let $R$ denote the resulting set. Then $|R|=k^d/e$.

We claim that the values $(X_{\bm r})_{\bm r\in R}$ together with the difference array $\bm\Delta=T_{\bm a}\bm X-\bm X$ uniquely determine $\bm X$. Indeed, suppose that $\bm X$ and $\bm Y$ have the same values on $R$ and satisfy $T_{\bm a}\bm X-\bm X= T_{\bm a}\bm Y-\bm Y$. Then $\bm Z=\bm X-\bm Y$ satisfies $T_{\bm a}\bm Z=\bm Z$ and $Z_{\bm r}=0$ for every $\bm r\in R$. Hence $\bm Z$ is constant on each cycle of $T_{\bm a}$ and vanishes at its representative. Therefore, $\bm Z=\bm 0$ and $\bm X=\bm Y$.

There are $q^{k^d/e}$ possible choices for $(X_{\bm r})_{\bm r\in R}$. Moreover, if $d_H(\bm X,T_{\bm a}\bm X)\leq s$, then $\wt(\bm\Delta)=d_H(\bm X,T_{\bm a}\bm X)\leq s$. Thus the difference array has at most $V_q(k^d,s)$ possible values. The claimed bound follows.
\end{proof}

For $\boldsymbol a=(a_1,\dots,a_d)\in\mathbb Z_k^d$, the translation $T_{\boldsymbol a}$ has order $k/ \gcd(k,a_1,\dots,a_d)$.
Thus every translation order divides $k$. Conversely, if $e\mid k$, then taking $\boldsymbol a=(k/e,0,\dots,0)$ gives a translation of order $e$. Hence the possible translation orders are exactly the divisors of $k$.

Now fix a divisor $e\mid k$. We count the translations of order $e$. Since $T_{\boldsymbol a}$ is determined by $\boldsymbol a$, this is equivalent to counting $\boldsymbol a\in\mathbb Z_k^d$ such that
\[
\frac{k}{\gcd(k,a_1,\dots,a_d)}=e.
\]
Every such $\boldsymbol a$ can be written uniquely as
\[
\boldsymbol a=\frac{k}{e}\boldsymbol u,
\qquad
\boldsymbol u=(u_1,\dots,u_d)\in\mathbb Z_e^d,
\]
and the condition becomes
\[
\gcd(e,u_1,\dots,u_d)=1.
\]
Thus it suffices to count the primitive vectors $\boldsymbol u\in\mathbb Z_e^d$ satisfying $\gcd(e,\boldsymbol u)=1$.
Such a vector fails to be primitive precisely when all its coordinates are divisible by some prime divisor $p$ of $e$. By inclusion-exclusion, the number of primitive vectors is
\[
J_d(e)
:=
e^d
\prod_{\substack{p\mid e,\,p\text{ prime}}}
\left(1-p^{-d}\right),
\]
the $d$-th Jordan totient function. Therefore, the number of translations $T_{\boldsymbol a}$ of order $e$ is also $J_d(e)$.

\begin{theorem}\label{thm:translation-packing}
We have
\begin{align*}
\Acpc_q(d,k,t)
&\ge
\frac{
q^{k^d}
-
V_q(k^d,2t)
\displaystyle\sum_{\substack{e\mid k,~e>1}}
J_d(e)q^{k^d/e}
}{
k^dV_q(k^d,2t)
}\\
&\ge
\frac{
q^{k^d}
-
(k^d-1)q^{k^d/2}V_q(k^d,2t)
}{
k^dV_q(k^d,2t)
}.
\end{align*}
Moreover,
\begin{equation*}
\Acpc_q(d,k,t)
\le
\frac{q^{k^d}}{k^dV_q(k^d,t)}.
\end{equation*}
\end{theorem}

\begin{proof}
We first prove the lower bounds. Let $\mathcal S$ be the set of arrays $\bm X$ such that $d_H(\bm X,T_{\bm a}\bm X)>2t$ for all $\bm a\in\Z_k^d\setminus \{\bm{0}\}$. If $\bm a$ has order $e$, then Lemma~\ref{lem:near-periodic} gives
\[
\left|
\left\{
\bm X:
d_H(\bm X,T_{\bm a}\bm X)\le2t
\right\}
\right|
\le
q^{k^d/e}V_q(k^d,2t).
\]
Since there are $J_d(e)$ translations of order $e$, the union bound yields
\[
|\mathcal S|
\ge
q^{k^d}- 
\sum_{\substack{e\mid k,~e>1}}
J_d(e)\, q^{k^d/e} V_q(k^d,2t).
\]

We then apply a greedy packing argument to $\mathcal S$. While arrays remain, select an array $\bm X$, include it in a set $\cC$, and remove every remaining array $\bm Y$ satisfying $d_H(\bm Y,T_{\bm a}\bm X)\le 2t$ for some $\bm a\in\Z_k^d$. For each $\bm a$, there are $V_q(k^d,2t)$ arrays within distance $2t$ of $T_{\bm a}\bm X$. Hence each selection removes at most $k^dV_q(k^d,2t)$ arrays, so
\[
|\cC|\ge\frac{|\mathcal S|}{k^dV_q(k^d,2t)}.
\]
By construction, every selected array is farther than $2t$ from each of its nontrivial translates, and every later selected array is farther than $2t$ from every translate of each earlier selected array. Thus $\cC$ is a $t$-substitution-correcting $d$D-CPC. Combining the two bounds gives the first lower bound in the theorem.

For the second lower bound, note that every nonzero translation has order $e>1$, so $q^{k^d/e}\le q^{k^d/2}$. Moreover, $\sum_{e\mid k}J_d(e)=k^d$, since this sum counts all translations $T_{\bm a}$. Hence $\sum_{e\mid k,~e>1}J_d(e)=k^d-1$. Substituting these relations into the first lower bound yields the second lower bound.

For the upper bound, we use a sphere-packing argument. Let $\cC$ be a $t$-substitution-correcting $d$D-CPC. The radius-$t$ Hamming balls centered at all translates $T_{\bm a}\bm X$ with $\bm X\in\cC$ and $\bm a\in\Z_k^d$ are pairwise disjoint. Therefore,
\[
|\cC| \cdot k^d \cdot V_q(k^d,t)\le q^{k^d},
\]
which gives the upper bound. This completes the proof.
\end{proof}

\begin{corollary}\label{cor:redundancy-noisy-CPC}
For fixed $q\geq 2$, $d\geq 1$, and $t\geq 1$, as $k\to\infty$, the optimal redundancy of $t$-substitution-correcting $d$D-CPCs ranges from $(t+1)d\log_q k+O(1)$ to $(2t+1)d\log_q k+O(1)$.
\end{corollary}

\begin{proof}
For fixed $r$, the volume of a Hamming ball satisfies $V_q(k^d,r)=\frac{(q-1)^r}{r!}k^{rd}\big(1+o(1)\big)$.
Thus, for fixed $q$, $d$, and $t$, as $k\to\infty$, we have $(k^d-1)q^{-k^d/2}V_q(k^d,2t)=o(1)$.
Applying Theorem~\ref{thm:translation-packing}, we obtain
\begin{align*}
k^d-\log_q\Acpc_q(d,k,t)
&\geq d\log_q k+\log_q V_q(k^d,t)
= (t+1)d\log_q k+O(1),\\
k^d-\log_q\Acpc_q(d,k,t)
&\leq d\log_q k+\log_q V_q(k^d,2t)+o(1)
= (2t+1)d\log_q k+O(1).
\end{align*}
This completes the proof.
\end{proof}

\section{Construction of Noiseless Two-Dimensional CPCs}\label{sec:noiseless}

In this section, we consider the two-dimensional case, which serves as a motivating special case for the higher-dimensional constructions developed in subsequent sections. We present three constructions for CPCs.

\subsection{Row-Column Marker Based Constructions}\label{subsec:marker}

The first construction reserves row~$0$ and column~$0$ as all-zero markers. Every other row and column must contain a nonzero symbol, so the two marker
lines remain unique after a cyclic translation.

Let $\cG^{\mathrm{rc}}_{k-1}$ be the set of
$(k-1)\times(k-1)$ arrays in which every row and column contains a nonzero
symbol. For $\bm U\in\cG^{\mathrm{rc}}_{k-1}$, let
$\bm B^{\mathrm{rc}}(\bm U)$ be the $k\times k$ array whose row~$0$ and
column~$0$ are zero and whose internal block is $\bm U$. Define
\[
\cB^{\mathrm{rc}}_{2,k}
=
\{\bm B^{\mathrm{rc}}(\bm U):\bm U\in\cG^{\mathrm{rc}}_{k-1}\}.
\]

\begin{lemma}\label{lem:unique-zero-2d}
For every $\bm B^{\mathrm{rc}}(\bm U)\in\cB^{\mathrm{rc}}_{2,k}$,
row~$0$ and column~$0$ are the unique all-zero row and column.
\end{lemma}

\begin{proof}
The two marker lines are zero by construction. Every other row and column
contains a row or column of $\bm U$, and is therefore nonzero.
\end{proof}

\begin{theorem}\label{thm:rc2-noiseless}
The code $\cB^{\mathrm{rc}}_{2,k}$ is a $2$D-CPC. For fixed $q$, as $k\rightarrow \infty$, its redundancy $\rho(\cB^{\mathrm{rc}}_{2,k})\leq2k-1+o(1)$.
\end{theorem}

\begin{proof}
We establish the CPC property by providing an efficient decoder. Suppose $\bm{Z}=T_{a,b}\bm{B}^{\mathrm{rc}}(\bm{U})$ for some $\bm{U}\in\cG^{\mathrm{rc}}_{k-1}$ and
$(a,b)\in\Z_k^2$. By Lemma~\ref{lem:unique-zero-2d}, the unique all-zero row and column of $\bm{Z}$ occur at indices $-a$ and $-b$, respectively. These two indices give the inverse translation $T_{-a,-b}$. Applying it to $\bm{Z}$ recovers $\bm{B}^{\mathrm{rc}}(\bm{U})$.
Hence, $\cB^{\mathrm{rc}}_{2,k}$ is a $2$D-CPC.

For the redundancy bound, choose
the internal block $\bm{U}$ uniformly at random. A fixed row or
column is all zero with probability $q^{-(k-1)}$. The union bound
over the $2(k-1)$ internal rows and columns gives
\[
\Pr\!\left(
\bm{U}\notin\cG^{\mathrm{rc}}_{k-1}
\right)
\le
2(k-1)q^{-(k-1)}.
\]
Since there are $q^{(k-1)^2}$ internal arrays, we can compute 
\[
|\cB^{\mathrm{rc}}_{2,k}|
\geq
q^{(k-1)^2}\bigl(1-2(k-1)q^{-(k-1)}\bigr).
\]
As $k\rightarrow \infty$, we have $2(k-1)q^{-(k-1)}=o(1)$, and therefore
\[
\rho\!\left(\cB^{\mathrm{rc}}_{2,k}\right)
\leq k^2-(k-1)^2+o(1)= 2k-1+o(1).
\]
This completes the proof.
\end{proof}

We now give a simple systematic encoder and its decoder. Set row $0$ and column $0$ to zero, set the internal main diagonal to $1$, and store the message in the remaining $(k-1)(k-2)$ internal positions in a fixed order. This gives an injective encoder for a subcode of $\mathcal B^{\mathrm{rc}}_{2,k}$ with redundancy $k^2-(k-1)(k-2)=3k-2$. To decode a translated array, first locate its unique zero row and zero column, then invert the translation, and finally read the message from the non-diagonal internal positions. Both encoding and decoding take $O(k^2)$ time.

\subsection{Row-Anchor Based Construction}\label{subsec:row-anchor}

The second construction uses one row as an anchor. The anchor belongs to a
one-dimensional CPC, while every other row lies outside all cyclic shifts of
that CPC. The anchor position gives the row shift, and the one-dimensional
CPC decoder gives the column shift.

Let $\cA\subseteq\Sigma_q^k$ be a one-dimensional CPC of size $s$, and let $T_b$ denote the cyclic translation by $b$ positions. Define the \emph{translation closure} of $\cA$ by
\[
\cU(\cA)
=
\left\{
T_b\bm{u}:
\bm{u}\in\cA,\ b\in\Z_k
\right\}.
\]
Since $\cA$ is a one-dimensional CPC, $|\cU(\cA)|=ks$.
Define the code
\[
\cB^{\mathrm{row}}_{2,k}(\cA)
=
\left\{
\bm{X}\in\Sigma_q^{k\times k}: \operatorname{row}_0(\bm{X})\in\cA,\,
\operatorname{row}_i(\bm{X})\notin\cU(\cA) \text{ for } i\in \Z_k \setminus \{0\}
\right\},
\]
where $\operatorname{row}_i(\bm{X})$ denotes the $i$-th row of $\bm{X}$.

\begin{theorem}\label{thm:row-anchor}
The code $\cB^{\mathrm{row}}_{2,k}(\cA)$ is a $2$D-CPC of size $s(q^k-ks)^{k-1}$.
Moreover, for fixed integer $q\geq 2$, as $k\to\infty$, there exists a choice for $\cA$ such that $\rho\big(\cB^{\mathrm{row}}_{2,k}(\cA)\big)=2\log_q k+\log_q e+o(1)$.
\end{theorem}

\begin{proof}
We establish the CPC property by describing a decoder. Suppose $\bm{Z}=T_{a,b}\bm{X}$ for some $\bm{X}\in\cB^{\mathrm{row}}_{2,k}(\cA)$ and $(a,b)\in\Z_k^2$. Since $\cU(\cA)$ is closed under cyclic translations, a row of $\bm{Z}$ belongs to $\cU(\cA)$ if and only if the corresponding row of $\bm{X}$ belongs to $\cU(\cA)$. By construction, row~$0$ is the unique row of $\bm{X}$ in $\cU(\cA)$. Hence $\bm{Z}$ has a unique such row. If its index is $i_0$, then $i_0+a\equiv0\pmod{k}$, so $a\equiv-i_0\pmod{k}$. Thus the row translation is uniquely determined.
Moreover, $\operatorname{row}_{i_0}(\bm{Z})=T_b(\operatorname{row}_0(\bm{X}))$. Since $\operatorname{row}_0(\bm{X})\in\cA$ and $\cA$ is a one-dimensional CPC, its decoder recovers the shift $b$. Hence $\bm{Z}$ uniquely determines the pair $(a,b)$, and applying the inverse translation recovers $\bm{X}$. This implies that  $\cB^{\mathrm{row}}_{2,k}(\cA)$ is a $2$D-CPC.

For the code size, row~$0$ has $s$ choices, while each of the remaining $k-1$ rows has $q^k-ks$ choices. Hence $\big|\cB^{\mathrm{row}}_{2,k}(\cA)\big|=s(q^k-ks)^{k-1}$.
By Corollary~\ref{cor:noiseless-cpc-asymptotic}, a largest one-dimensional CPC has size $\big(1-o(1)\big)q^k/k$. We can therefore choose a one-dimensional CPC $\cA$ of size
$s= \bigl(1-o(1)\bigr)q^k/k^2$.
With this choice,
\[
(q^k-ks)^{k-1}
=
q^{k(k-1)}
\left(1-\frac{1-o(1)}{k}\right)^{k-1}
=
\frac{q^{k(k-1)}}{e}\bigl(1-o(1)\bigr),
\]
where the last step follows from $\left(1-\frac{1}{k}\right)^{k}\to \frac{1}{e}$. Therefore,
\begin{equation}\label{eq:2d-row-anchor-size}
|\cB^{\mathrm{row}}_{2,k}(\cA)|
=
\frac{q^k}{k^2}\big(1-o(1)\big)\cdot \frac{q^{k(k-1)}}{e}\bigl(1-o(1)\bigr)
=
\frac{q^{k^2}}{e\,k^2}\bigl(1-o(1)\bigr).
\end{equation}
It follows that
\[
\rho\!\left(\cB^{\mathrm{row}}_{2,k}(\cA)\right)
=
2\log_q k+\log_q e+o(1),
\]
which completes the proof.
\end{proof}

We next give an explicit encoder and decoder for
$\cB^{\mathrm{row}}_{2,k}(\cA)$ with redundancy
$2\log_q k+O(1)$, which is optimal up to an additive constant.
For the anchor row, we use the construction of mutually uncorrelated
codes\footnote{One-dimensional mutually uncorrelated codes, also known as non-overlapping codes or cross-bifix-free codes, are defined by the property that no nonempty proper prefix of any codeword is a suffix of any codeword. It follows from the definition that such codes are $1$D-CPCs. In higher dimensions, the relation depends on the overlap convention and requires a separate argument, which we do not pursue here. For multidimensional non-overlapping codes, see \cite{CaiWangFeng2025,QinZhang2026}.}
from~\cite{LevyYaakobi2019}.

\begin{lemma}\label{lem:non-overlapping}
Let
\[
    \cA=
    \left\{
        0^f1\bm{u}1\in \Sigma_q^k:
        \bm{u}\in\Sigma_q^{k-f-2}
        \text{ does not contain }0^f
    \right\}
\]
be the mutually uncorrelated code defined in~\cite{LevyYaakobi2019}.
Then $\cA$ is also a $1$D-CPC. Moreover, if
$\bm{x}\in\Sigma_q^k$ contains no cyclic occurrence of $0^f$,
then $\bm{x}\notin\cU(\cA)$.
\end{lemma}

\begin{proof}
Every sequence in $\cA$ contains exactly one cyclic occurrence of
$0^f$, namely its prefix. Suppose $\bm{Z}=T_a\bm{X}$ for some
$\bm{X}\in\cA$ and $a\in\Z_k$. If the unique cyclic occurrence
of $0^f$ in $\bm{Z}$ starts at position $i_0$, then $a\equiv-i_0\pmod{k}$. Applying the inverse shift recovers $\bm{X}$. Thus $\cA$ is a
$1$D-CPC.

Since every sequence in $\cU(\cA)$ contains a cyclic occurrence of
$0^f$, any $\bm{x}\in\Sigma_q^k$ with no such occurrence cannot
belong to $\cU(\cA)$. This completes the proof.
\end{proof}

For $\lceil\log_q k\rceil+1\le f\le k-3$, the sequence replacement technique of \cite[Section III-B]{LevyYaakobi2019} yields an encoder $\mathrm{ENC}^{\cA}:\Sigma_q^{k-f-3}\to\cA$ together with its inverse $\mathrm{DEC}^{\cA}$, both running in $O(k)$ time. To complete the encoder, it remains to encode the remaining $k-1$ rows so that none contains a cyclic occurrence of $0^f$.

Let $N$ be a positive multiple of $k$. We regard a word in
$\Sigma_q^N$ as an $(N/k)\times k$ array filled row by row and set
\[
f=\left\lceil\log_q N\right\rceil+1.
\]
We assume that $f\leq k$. This holds for all sufficiently large $k$ in
the applications below. For $p\in\{0,\ldots,N-1\}$, let
$\operatorname{Rep}_{q,f-1}(p)$ denote the length-$(f-1)$
$q$-ary representation of $p$. 

A \emph{row-cyclic occurrence} of $0^f$ starting at $i=rk+c$ consists of
the positions
\[
rk+\bigl((c+\ell)\bmod k\bigr),\qquad 0\leq\ell<f.
\]
When $c\leq k-f$, these positions form one contiguous block. When
$c>k-f$, set $b=f-(k-c)$, the occurrence consists of the last
$k-c$ positions and the first $b$ positions of row $r$. In the
latter case, deleting the occurrence means that these two blocks
are deleted simultaneously and the remaining symbols retain their
relative order.

\begin{algorithm}[t]
\caption{Row-cyclic RLL encoder $\mathrm{ENC}^{\mathrm{RLL}}_{N,k}$}
\label{alg:row-cyclic-encoder}
\begin{algorithmic}[1]
\Require $\bm{x}\in\Sigma_q^{N-1}$, where $k\mid N$
\Ensure $\bm{y}\in\Sigma_q^N$ with no row-cyclic occurrence of $0^f$, where  $f=\lceil\log_qN\rceil+1\leq k$
\State $\bm{y}\gets\bm{x}0$, $i\gets0$
\While{$i<N$}
    \State Find the first row-cyclic $0^f$ whose starting position
           $i'\geq i$
    \If{no such occurrence exists}
        \State \Return $\bm{y}$
    \EndIf
    \State Write $i'=rk+c$
    \If{$c\leq k-f$}
        \State Delete $\bm{y}[i',i'+f-1]$
        \State $i\gets i'$
    \Else
        \State $b\gets f-(k-c)$
        \State Delete
               $\bm{y}_{[rk,rk+b-1]}$ and
               $\bm{y}_{[i',rk+k-1]}$
        \State $i\gets i'-b$
    \EndIf
    \State Append $\operatorname{Rep}_{q,f-1}(i')\circ1$ to
           $\bm{y}$
\EndWhile
\State \Return $\bm{y}$
\end{algorithmic}
\end{algorithm}

\begin{algorithm}[t]
\caption{Row-cyclic RLL decoder $\mathrm{DEC}^{\mathrm{RLL}}_{N,k}$}
\label{alg:row-cyclic-decoder}
\begin{algorithmic}[1]
\Require $\bm{y}$ produced by
         Algorithm~\ref{alg:row-cyclic-encoder}
\Ensure $\bm{x}\in\Sigma_q^{N-1}$
\While{the last symbol of $\bm{y}$ is $1$}
    \State Read the last $f$ symbols as
           $\operatorname{Rep}_{q,f-1}(i')\circ1$ and delete them
    \State Write $i'=rk+c$
    \If{$c\leq k-f$}
        \State Insert $0^f$ at position $i'$
    \Else
        \State $b\gets f-(k-c)$
        \State Insert $0^b$ at position $rk$
        \State Insert $0^{k-c}$ at position $rk+c$
    \EndIf
\EndWhile
\State Delete the final symbol $0$
\State \Return the remaining word as $\bm{x}$
\end{algorithmic}
\end{algorithm}

\begin{theorem}\label{thm:row-cyclic-rll}
Algorithms~\ref{alg:row-cyclic-encoder} and
\ref{alg:row-cyclic-decoder} are mutually inverse. The encoder output
contains no row-cyclic occurrence of $0^f$. For fixed $q$, both
algorithms can be implemented in $O(N\log^2 N)$ time.
\end{theorem}

\begin{proof}
We first justify the scanning invariant: at the beginning of each
iteration, no row-cyclic occurrence of $0^f$ starts before $i$. Let the
first occurrence found start at $i'=rk+c\geq i$.
\begin{itemize}
    \item If $c\leq k-f$, all rows preceding row $r$ are unchanged. A window in row $r$ starting before $c$ either is unchanged or spans the gap left by the deletion. If such a spanning window were all zero, restoring the deleted zeros would yield an all-zero length-$f$ window at the same earlier position in the previous row, contradicting the choice of $i'$. Thus scanning may resume at $i'$.

\item If $c>k-f$, set $b=f-(k-c)$. The retained portion of row $r$ shifts left by $b$ positions, and $i'-b=rk+k-f\geq rk$. For a new window starting at $rk+s$ with $s<k-f$, its retained symbols correspond to the old row window starting at $s+b<c$. Any additional symbols needed to complete the length-$f$ cyclic window at $s+b$ lie in the deleted zero suffix or prefix. Hence, if such a new window were all zero, the corresponding old window would also have been all zero at an earlier position. Since rows before row $r$ are unchanged, it suffices to resume scanning from $i'-b$.
\end{itemize}
The invariant is preserved in both cases, and the restart is at most
$f-1$ positions before the selected starting position.

Each replacement removes only zeros and appends a pointer ending in
$1$. Hence the number of nonzero symbols increases by at least one.
Since the word always has length $N$, at most $N$ replacements can
occur. The encoder therefore terminates, and its stopping condition
ensures that the output contains no row-cyclic occurrence of $0^f$.

We next verify the decoder. The initial word $\bm{x}0$ ends in $0$,
whereas the word after every replacement ends in $1$. Thus the final
$f$ symbols identify the pointer appended in the most recent
replacement. In the noncrossing case, the decoder simply reinserts
$0^f$ at $i'$. In the crossing case, it first inserts
$b=f-(k-c)$ zeros at $rk=i'-c$ and then inserts the remaining
$k-c$ zeros at $i'$.
Consequently, each decoding step exactly reverses the most recent
replacement. Repeating this procedure recovers $\bm{x}0$, after
which the final zero is removed.

Finally, let $R\le N$ be the number of replacements. The total number of scanned starting positions is $O(N+Rf)=O(Nf)$. Since checking each window takes $O(f)$ time, the total scanning time is $O(Nf^2)$. For each replacement, encoding or decoding the position $i<N$ amounts to converting between $i$ and its $q$-ary representation, which takes $O(f)$ time. All replacements thus take $O(Rf)=O(Nf)$ time. The total encoding and decoding time is therefore $O(Nf^2)=O(N\log^2 N)$. This completes the proof.
\end{proof}

For the two-dimensional construction, take $N=k(k-1)$ and assume
$f+3\leq k$. Combining the encoder and decoder for the anchor row
with Algorithms~\ref{alg:row-cyclic-encoder} and
\ref{alg:row-cyclic-decoder} for the remaining $k-1$ rows gives an
explicit subcode of $\cB^{\mathrm{row}}_{2,k}(\cA)$. The anchor row
uses $f+3$ redundant symbols, while the remaining rows jointly use
one redundant symbol. Hence the overall redundancy is
\[
    f+3+1
    =
    \left\lceil\log_q k(k-1)\right\rceil+5
    =
    2\log_q k+O(1),
\]
which is optimal up to an additive constant.
For fixed $q$, encoding and decoding take $O(k^2\log^2 k)$ time.

\subsection{Moment-Syndrome Based Construction}\label{subsec:moment-2d}

The third construction uses no marker. Instead, each array is equipped with a pair of moments that transform predictably under translation. By prescribing their target values, we recover the shift by solving linear congruences, without reserving any rows, columns, or symbols.

For $\bm{X}\in\Sigma_q^{k\times k}$, define its row and column moments by
\begin{align*}
R(\bm{X})=\sum_{(i,j)\in\supp(\bm{X})} i \pmod{k}, \qquad
C(\bm{X})=\sum_{(i,j)\in\supp(\bm{X})} j \pmod{k}.
\end{align*}
For a prescribed syndrome $(r_0,c_0)\in\Z_k^2$, the corresponding \emph{moment-syndrome code} is defined as
\[
\cB^{\mathrm{ms}}_{2,k}(r_0,c_0)
=
\left\{
\bm{X}\in\Sigma_q^{k\times k}:
\gcd(\wt(\bm{X}),k)=1,\;
R(\bm{X})=r_0,\;
C(\bm{X})=c_0
\right\}.
\]

\begin{theorem}\label{thm:moment-2d}
For every $(r_0,c_0)\in\Z_k^2$, the code
$\cB^{\mathrm{ms}}_{2,k}(r_0,c_0)$ is a $2$D-CPC of size
\begin{equation*}
\left|\cB^{\mathrm{ms}}_{2,k}(r_0,c_0)\right|
=
\frac{1}{k^2}
\sum_{\substack{0\le \ell\le k^2\\ \gcd(\ell,k)=1}}
\binom{k^2}{\ell}(q-1)^\ell.
\end{equation*}
\end{theorem}

\begin{proof}
We establish the CPC property by providing an efficient decoder. Suppose $\bm{Z}=T_{a,b}\bm{X}$ for some $\bm{X}\in\cB^{\mathrm{ms}}_{2,k}(r_0,c_0)$ and $(a,b)\in\Z_k^2$. Translation preserves Hamming weight and transforms the moments as
\begin{align*}
R(\bm{Z}) &\equiv R(\bm{X}) - a\,\wt(\bm{X}) \pmod{k},\\
C(\bm{Z}) &\equiv C(\bm{X}) - b\,\wt(\bm{X}) \pmod{k}.
\end{align*}
Since the weight is invertible modulo $k$, the translation $(a,b)$ can be recovered by
\begin{align*}
a &\equiv \bigl(r_0-R(\bm{Z})\bigr)\cdot \wt(\bm{Z})^{-1} \pmod{k},\\
b &\equiv \bigl(c_0-C(\bm{Z})\bigr)\cdot \wt(\bm{Z})^{-1} \pmod{k}.
\end{align*}
Applying the inverse translation gives $T_{-a,-b}\bm{Z}=\bm{X}$. Hence $\cB^{\mathrm{ms}}_{2,k}(r_0,c_0)$ is a $2$D-CPC.

For the code size, let
\[
\cS = \left\{ \bm{X}\in\Sigma_q^{k\times k}: \gcd(\wt(\bm{X}),k)=1 \right\}.
\]
The set $\cS$ is invariant under translations. Fix $\bm{X}\in\cS$. As $a$ ranges over $\Z_k$, the quantity $R(\bm{X})-a\,\wt(\bm{X})$ runs through all residues modulo $k$ exactly once. Similarly, as $b$ ranges over $\Z_k$, $C(\bm{X})-b\,\wt(\bm{X})$ does the same. Consequently, the map
\[
(a,b) \longmapsto \bigl( R(T_{a,b}\bm{X}),\, C(T_{a,b}\bm{X}) \bigr)
\]
is a bijection from $\Z_k^2$ onto $\Z_k^2$. In particular, for each prescribed syndrome $(r_0,c_0)$, there is a unique pair $(a,b)\in\Z_k^2$ such that
\[
R(T_{a,b}\bm{X})=r_0, \qquad C(T_{a,b}\bm{X})=c_0.
\]
Thus each translation necklace within $\cS$ has size $k^2$ and contains exactly one representative from $\cB^{\mathrm{ms}}_{2,k}(r_0,c_0)$. It follows that
\[
\left|\cB^{\mathrm{ms}}_{2,k}(r_0,c_0)\right|
=
\frac{|\cS|}{k^2}
=
\frac{1}{k^2}
\sum_{\substack{0\le\ell\le k^2\\\gcd(\ell,k)=1}}
\binom{k^2}{\ell}(q-1)^\ell.
\]
This completes the proof.
\end{proof}

We next estimate the redundancy of
$\cB^{\mathrm{ms}}_{2,k}(r_0,c_0)$ when the side length $k$ is prime.
We use the standard roots-of-unity filter.

\begin{lemma}[Roots-of-unity filter]\label{lem:roots-unity-filter}
Let $F(z)=\sum_{\ell=0}^{n}a_\ell z^\ell$ and
$\omega=e^{2\pi i/k}$ (here $i$ denotes the imaginary unit). Then
\[
\sum_{\substack{0\le \ell\le n\\ k\mid\ell}}a_\ell
=
\frac{1}{k}\sum_{j=0}^{k-1}F(\omega^j).
\]
\end{lemma}

\begin{proof}
For every integer $\ell$,
\[
\frac{1}{k}\sum_{j=0}^{k-1}\omega^{j\ell}
=
\begin{cases}
1, & k\mid \ell,\\
0, & k\nmid \ell.
\end{cases}
\]
Therefore,
\[
\begin{aligned}
\sum_{\substack{0\le \ell\le n\\ k\mid\ell}}a_\ell
&=
\sum_{\ell=0}^{n}a_\ell
\left(\frac{1}{k}\sum_{j=0}^{k-1}\omega^{j\ell}\right)
=
\frac{1}{k}\sum_{j=0}^{k-1}
\sum_{\ell=0}^{n}a_\ell\omega^{j\ell}
=
\frac{1}{k}\sum_{j=0}^{k-1}F(\omega^j).
\end{aligned}
\]
This completes the proof.
\end{proof}

\begin{lemma}\label{lem:MSC-2d}
For fixed integer $q\geq 2$, as $k\to\infty$ through the primes,
\[
\rho\!\left(\cB^{\mathrm{ms}}_{2,k}(r_0,c_0)\right)
=
2\log_q k+o(1),
\]
which matches the optimal redundancy up to a
vanishing additive term.
\end{lemma}

\begin{proof}
Since $k$ is prime, $\gcd(\ell,k)=1$ if and only if $k\nmid\ell$.
Let
\[
S_k=
\sum_{\substack{0\le \ell\le k^2\\ k\mid\ell}}
\binom{k^2}{\ell}(q-1)^\ell.
\]
Then
$\big|\cB^{\mathrm{ms}}_{2,k}(r_0,c_0)\big|
=(q^{k^2}-S_k)/k^2$.
To estimate $S_k$, apply Lemma~\ref{lem:roots-unity-filter} to
$F(z)=(1+(q-1)z)^{k^2}$. With $\omega=e^{2\pi i/k}$, we obtain
\[
S_k
= \frac{1}{k}\sum_{j=0}^{k-1}
\bigl(1+(q-1)\omega^j\bigr)^{k^2}
=
\frac{q^{k^2}}{k}
+
\frac{1}{k}\sum_{j=1}^{k-1}
\bigl(1+(q-1)\omega^j\bigr)^{k^2}.
\]

For $1\le j\le k-1$,
\begin{align*}
\left|1+(q-1)\omega^j\right|^2
&=
q^2-2(q-1)\left(1-\cos\frac{2\pi j}{k}\right)
=
q^2-4(q-1)\sin^2\frac{\pi j}{k}.
\end{align*}
Let $m_j=\min\{j,k-j\}$. Then $1\le m_j\le k/2$ and
$\sin(\pi j/k)=\sin(\pi m_j/k)$. Since
$\sin x\ge 2x/\pi$ for $0\le x\le\pi/2$, we have
$\sin(\pi j/k)\ge 2m_j/k$. Hence
\begin{equation}\label{eq:filter}
\frac{\left|1+(q-1)\omega^j\right|^2}{q^2}
\le
1-\frac{16(q-1)m_j^2}{q^2k^2}.
\end{equation}
Using $1-x\le e^{-x}$, it follows that
\[
\left|
\frac{1+(q-1)\omega^j}{q}
\right|^{k^2}
\le
\exp\!\left(
-\frac{8(q-1)}{q^2}m_j^2
\right).
\]
Since each positive integer $m$ occurs as $m_j$ for at most two values of $j$, we have
\[
\sum_{j=1}^{k-1}
\left|
\frac{1+(q-1)\omega^j}{q}
\right|^{k^2}
\le
2\sum_{m=1}^{\infty}
\exp\!\left(
-\frac{8(q-1)}{q^2}m^2
\right)
=
O(1),
\]
where the last equality holds since the series is a convergent Gaussian series.
By the triangle inequality,
\begin{align*}
\left|
S_k-\frac{q^{k^2}}{k}
\right|
&\le
\frac{q^{k^2}}{k}
\sum_{j=1}^{k-1}
\left|
\frac{1+(q-1)\omega^j}{q}
\right|^{k^2}
=
O\!\left(\frac{q^{k^2}}{k}\right).
\end{align*}
In particular, $S_k=O(q^{k^2}/k)$ and $\big|\cB^{\mathrm{ms}}_{2,k}(r_0,c_0)\big|
=(q^{k^2}-S_k)/k^2=\big(1-O(1/k)\big)q^{k^2}/k^2$.
It follows that
\begin{align*}
\rho\!\left(\cB^{\mathrm{ms}}_{2,k}(r_0,c_0)\right)
&=
2\log_q k
-\log_q\big(1-O(1/k)\big)=
2\log_q k+o(1),
\end{align*}
which completes the proof.
\end{proof}

\section{Construction of Noiseless Multidimensional CPCs}\label{sec:noiseless_multi}

This section extends the three two-dimensional constructions to fixed dimension
$d\geq2$.

\subsection{Marker-Hyperplane Based Construction}

The row-column marker becomes a set of $d$ coordinate hyperplanes. We set the
hyperplane at index~$0$ in each direction to zero and require every other
coordinate hyperplane to contain a nonzero entry.

Let
$\cG_{d,k-1}\subseteq\Sigma_q^{\Z_{k-1}^d}$ be the set of internal
arrays with no all-zero coordinate hyperplane. For $\bm{U}\in\cG_{d,k-1}$, define $\bm{B}(\bm{U})\in\Sigma_q^{\Z_k^d}$ by placing $\bm{U}$ on $([k]^+)^d$ and setting all remaining entries to $0$. The code is then defined as
\begin{equation*}
\cB^{\mathrm{mk}}_{d,k}
=
\left\{
\bm{B}(\bm{U}):
\bm{U}\in\cG_{d,k-1}
\right\}.
\end{equation*}

\begin{theorem}\label{thm:d-marker-noiseless}
The code $\cB^{\mathrm{mk}}_{d,k}$ is a $d$D-CPC. For fixed $d\ge2$ and $q$, as $k\rightarrow \infty$, $\rho\!\left(\cB^{\mathrm{mk}}_{d,k}\right)
\leq
k^d-(k-1)^d+o(1)$.
\end{theorem}

\begin{proof}
We establish the CPC property by providing an efficient decoder. Suppose $\bm{Z}=T_{\bm{a}}\bm{B}(\bm{U})$ for some $\bm{U}\in\cG_{d,k-1}$ and $\bm{a}=(a_1,\ldots,a_d)\in\Z_k^d$. In the original array $\bm{B}(\bm{U})$, the coordinate-$j$ hyperplane at index~$0$ is all zero for every $1\leq j\leq d$, while every other coordinate-$j$ hyperplane contains a nonzero internal hyperplane of $\bm{U}$ and is therefore nonzero. After translation by $\bm{a}$, the unique all-zero coordinate-$j$ hyperplane of $\bm{Z}$ occurs at index $i_j\equiv -a_j\pmod{k}$. Thus $\bm{i}=(i_1,\ldots,i_d)=-\bm{a}$ determines the inverse translation. Applying $T_{\bm{i}}$ recovers $\bm{B}(\bm{U})$, so $\cB^{\mathrm{mk}}_{d,k}$ is a $d$D-CPC.

For the size bound, choose $\bm{U}\in\Sigma_q^{\Z_{k-1}^d}$ uniformly at random. A fixed coordinate hyperplane contains $(k-1)^{d-1}$ entries and is all zero with probability $q^{-(k-1)^{d-1}}$. The union bound over the $d(k-1)$ internal hyperplanes gives
\[
\Pr\!\left(
\bm{U}\notin\cG_{d,k-1}
\right)
\le
d(k-1)q^{-(k-1)^{d-1}}.
\]
Since there are $q^{(k-1)^d}$ internal arrays, we can compute
\begin{align*}
\left| \cB^{\mathrm{mk}}_{d,k} \right|= \left|\cG_{d,k-1}\right|
&\ge
q^{(k-1)^d}
\left(
1-d(k-1)q^{-(k-1)^{d-1}}
\right).
\end{align*}
For fixed $d\ge2$ and $q$, as $k\to\infty$, we have $d(k-1)q^{-(k-1)^{d-1}}=o(1)$ and therefore
\[
\rho\!\left(\cB^{\mathrm{mk}}_{d,k}\right)
=k^d- \log_q\big|\cB^{\mathrm{mk}}_{d,k}\big|=
k^d-(k-1)^d+o(1).
\]
This completes the proof.
\end{proof}

As in the row-column marker based construction of Section~\ref{subsec:marker}, a systematic encoder for
$\cB^{\mathrm{mk}}_{d,k}$ can be obtained by fixing the entries on the
internal main diagonal to $1$. The resulting code
has redundancy $k^d-(k-1)^d+k-1$,
and admits explicit encoding and decoding algorithms, both with time
complexity $O(k^d)$.

\subsection{Row-Anchor Based Construction}
\label{subsec:row-anchor-multi}

We use the same row-anchor idea in $d$ dimensions. One length-$k$ row is the
anchor. Its row index gives the first $d-1$ shifts, and the
one-dimensional CPC decoder gives the last shift.

For $\bm{X}\in\Sigma_q^{\Z_k^d}$ and
$\bm{i}\in\Z_k^{d-1}$, define
\begin{equation}\label{eq:d-row-definition}
    \operatorname{row}_{\bm{i}}(\bm{X})
    =\bigl(X_{(\bm{i},j)}\bigr)_{j\in\Z_k}\in\Sigma_q^k.
\end{equation}
Let $\cA\subseteq\Sigma_q^k$ be a one-dimensional CPC of size $s$,
and let $\cU(\cA)$ be its translation closure, as defined in
Section~\ref{subsec:row-anchor}. Recall that $|\cU(\cA)|=ks$.
Define the code
\begin{equation}\label{eq:d-row-anchor-code}
    \cB^{\mathrm{row}}_{d,k}(\cA)
    =\left\{
    \bm{X}\in\Sigma_q^{\Z_k^d}\;\middle|\;
    \begin{array}{l}
        \operatorname{row}_{\bm{0}}(\bm{X})\in\cA,\\
        \operatorname{row}_{\bm{i}}(\bm{X})\notin\cU(\cA)
        \text{ for }\bm{i}\in\Z_k^{d-1}\setminus\{\bm{0}\}
    \end{array}
    \right\}.
\end{equation}
For $d=2$, this is exactly the code in
Theorem~\ref{thm:row-anchor}.

\begin{theorem}\label{thm:row-anchor-multi}
The code $\cB^{\mathrm{row}}_{d,k}(\cA)$ is a $d$D-CPC of size
$s(q^k-ks)^{k^{d-1}-1}$. Moreover, for fixed integers $q\geq2$ and
$d\geq2$, as $k\to\infty$, there exists a choice of $\cA$ such that
$\rho\bigl(\cB^{\mathrm{row}}_{d,k}(\cA)\bigr)
=d\log_q k+\log_q e+o(1)$.
\end{theorem}

\begin{proof}
We establish the CPC property by describing a decoder. Suppose
$\bm{Z}=T_{(\bm{a},b)}\bm{X}$ for some
$\bm{X}\in\cB^{\mathrm{row}}_{d,k}(\cA)$ and
$(\bm{a},b)\in\Z_k^{d-1}\times\Z_k$. By the definition of translation,
\begin{equation}\label{eq:d-row-translation}
    \operatorname{row}_{\bm{i}}(\bm{Z})
    =T_b\bigl(\operatorname{row}_{\bm{i}+\bm{a}}(\bm{X})\bigr).
\end{equation}
Since $\cU(\cA)$ is closed under cyclic translations, a row of
$\bm{Z}$ belongs to $\cU(\cA)$ if and only if the corresponding
row of $\bm{X}$ belongs to $\cU(\cA)$. By construction,
row~$\bm{0}$ is the unique such row of $\bm{X}$. Hence $\bm{Z}$
has a unique row in $\cU(\cA)$. If its index is $\bm{i}_0$, then
$\bm{i}_0+\bm{a}\equiv \bm{0} \pmod{k}$, so
$\bm{a}\equiv -\bm{i}_0 \pmod{k}$. This gives the first $d-1$
translation coordinates. Also,
$\operatorname{row}_{\bm{i}_0}(\bm{Z})
=T_b\bigl(\operatorname{row}_{\bm{0}}(\bm{X})\bigr)$.
Since $\operatorname{row}_{\bm{0}}(\bm{X})\in\cA$ and $\cA$ is a
one-dimensional CPC, its decoder recovers $b$. Thus $\bm Z$ determines the full translation, and the inverse translation
recovers $\bm X$. This implies that
$\cB^{\mathrm{row}}_{d,k}(\cA)$ is a $d$D-CPC.

For the code size, row~$\bm{0}$ has $s$ choices, while each of the
remaining $k^{d-1}-1$ rows has $q^k-ks$ choices. Hence
$\bigl|\cB^{\mathrm{row}}_{d,k}(\cA)\bigr|
=s(q^k-ks)^{k^{d-1}-1}$.
By Corollary~\ref{cor:noiseless-cpc-asymptotic}, a largest
one-dimensional CPC has size $\big(1-o(1)\big)q^k/k$. We can therefore choose a subcode $\cA$ of size
$s=\big(1-o(1)\big)q^k/k^d$. With this choice,
\[
    (q^k-ks)^{k^{d-1}-1}
    =q^{k(k^{d-1}-1)}
    \left(1-\frac{1-o(1)}{k^{d-1}}\right)^{k^{d-1}-1}
    =\frac{q^{k(k^{d-1}-1)}}{e}(1-o(1)).
\]
Thus,
$\bigl|\cB^{\mathrm{row}}_{d,k}(\cA)\bigr|
=q^{k^d}\big(1-o(1)\big)/(e\,k^d)$ and
$\rho\bigl(\cB^{\mathrm{row}}_{d,k}(\cA)\bigr)
=d\log_q k+\log_q e+o(1)$. This completes the proof.
\end{proof}

We next give an explicit subcode with redundancy $d\log_q k+O(1)$. Set
$R=k^{d-1}-1$ and $N=Rk=k^d-k$, and fix the lexicographic ordering $\bm{i}_1,\ldots,\bm{i}_R$ of
$\Z_k^{d-1}\setminus\{\bm{0}\}$. Take
$f=\lceil\log_q N\rceil+1$ and assume $f+3\leq k$. For fixed $q$
and $d$, this holds for all sufficiently large $k$.

For the anchor row, we use the non-overlapping code $\cA$ from
Lemma~\ref{lem:non-overlapping}, with marker $0^f$. Its encoder
$\mathrm{ENC}^{\cA}\colon\Sigma_q^{k-f-3}\to\cA$ and decoder
$\mathrm{DEC}^{\cA}$ are the constituent algorithms used in
Section~\ref{subsec:row-anchor}, based on
\cite[Section III-B]{LevyYaakobi2019}.
For the remaining $R$ rows, we use
$\mathrm{ENC}^{\mathrm{RLL}}_{N,k}$ and
$\mathrm{DEC}^{\mathrm{RLL}}_{N,k}$ from
Theorem~\ref{thm:row-cyclic-rll}. These rows are encoded jointly,
while the cyclic constraint is imposed separately within each
length-$k$ row.

Given $\bm{x}_{\mathrm A}\in\Sigma_q^{k-f-3}$ and
$\bm{x}_{\mathrm R}\in\Sigma_q^{N-1}$, the encoder
$\mathrm{ENC}^{\mathrm{row}}_{d,k}$ first computes
$\bm{u}=\mathrm{ENC}^{\cA}(\bm{x}_{\mathrm A})$ and places
$\bm{u}$ in row~$\bm{0}$. It then computes
$\bm{y}=\mathrm{ENC}^{\mathrm{RLL}}_{N,k}(\bm{x}_{\mathrm R})$,
partitions $\bm{y}$ into $R$ consecutive blocks of length $k$,
and places the $j$th block in row~$\bm{i}_j$ for
$1\leq j\leq R$. Denote the resulting array by $\bm{X}$.
By Theorem~\ref{thm:row-cyclic-rll}, every non-anchor row contains
no cyclic occurrence of $0^f$. Lemma~\ref{lem:non-overlapping}
therefore implies that each of these rows lies outside
$\cU(\cA)$. Since row~$\bm{0}$ belongs to $\cA$, we have
$\bm{X}\in\cB^{\mathrm{row}}_{d,k}(\cA)$. Moreover, both
constituent encoders are injective, so
$\mathrm{ENC}^{\mathrm{row}}_{d,k}$ is an injective encoder
for a subcode of $\cB^{\mathrm{row}}_{d,k}(\cA)$.

The decoding procedure is similarly efficient. Given
$\bm{Z}=T_{(\bm{a},b)}\bm{X}$, the decoder first locates the unique
cyclic occurrence of $0^f$ among all rows of $\bm{Z}$. Suppose this
occurrence starts at position $c_0$ in row~$\bm{i}_0$. 
By
\eqref{eq:d-row-translation},
$\bm{a}\equiv-\bm{i}_0\pmod{k}$ and
$b\equiv-c_0\pmod{k}$. Applying the inverse translation recovers
the original array $\bm{X}$. The decoder then applies
$\mathrm{DEC}^{\cA}$ to row~$\bm{0}$ to recover
$\bm{x}_{\mathrm A}$. Finally, concatenating rows
$\bm{i}_1,\ldots,\bm{i}_R$ in the fixed order recovers $\bm{y}$,
and applying $\mathrm{DEC}^{\mathrm{RLL}}_{N,k}$ yields
$\bm{x}_{\mathrm R}$.

The anchor row uses $f+3$ redundant symbols, while the remaining
rows jointly use one redundant symbol. Hence the overall
redundancy is
\begin{equation}\label{eq:d-row-anchor-explicit-redundancy}
    f+4=\left\lceil\log_q(k^d-k)\right\rceil+5
    =d\log_q k+O(1).
\end{equation}
The marker can be located in $O(k^d)$ time.
Together with the constituent algorithms in
Section~\ref{subsec:row-anchor} and
Theorem~\ref{thm:row-cyclic-rll}, this gives an overall time
complexity of $O(k^d\log^2 k)$ for both encoding and decoding,
for fixed $q$ and $d$.

\begin{remark}
An alternative extension of the two-dimensional row-anchor
construction is to use a $(d-1)$-dimensional slice as the anchor.
The anchor slice is chosen from a $(d-1)$D-CPC, while every other
parallel slice is constrained to lie outside its translation closure.
After an unknown translation, the location of the anchor slice
determines one translation coordinate, and the $(d-1)$D-CPC decoder
recovers the remaining $d-1$ coordinates. By choosing the size of
the anchor code appropriately, the same counting argument gives
redundancy $d\log_q k+\log_q e+o(1)$. As the argument is similar, we omit the details.
\end{remark}

\subsection{Moment-Syndrome Based Construction}

The third construction uses no marker. Each array is equipped with a vector of moments computed from the positions of its nonzero symbols, which changes predictably under translation. By fixing the desired syndrome, we recover the shift by solving linear congruences. This generalizes the two-dimensional moment-syndrome code from Section~\ref{subsec:moment-2d}.

For $\bm{X}\in\Sigma_q^{\Z_k^d}$, recall that its \emph{support} is defined by $\supp(\bm{X}) = \big\{ \bm{i}\in\Z_k^d: X_{\bm{i}}\ne0 \big\}$. Let
\begin{align*}
S_j(\bm{X})&=\sum_{\bm{i}\in\supp(\bm{X})}i_j\pmod{k},
 \qquad 1\le j\le d.
\end{align*}
For a prescribed syndrome $\bm{\sigma}=(\sigma_1,\ldots,\sigma_d)\in\Z_k^d$, define the code
\begin{equation*}
 \cB^{\mathrm{ms}}_{d,k}(\bm{\sigma})
 =\left\{\bm{X}\in\Sigma_q^{\Z_k^d}:
 \gcd(\wt(\bm{X}),k)=1,\,
 S_j(\bm{X})=\sigma_j \text{ for } 1\le j\le d
 \right\}.
\end{equation*}

\begin{theorem}\label{thm:moment}
For every $\bm{\sigma}\in\Z_k^d$, the code $\cB^{\mathrm{ms}}_{d,k}(\bm{\sigma})$ is a $d$D-CPC of size
\begin{equation*}
 \left|\cB^{\mathrm{ms}}_{d,k}(\bm{\sigma})\right|
 =\frac{1}{k^d}
 \sum_{\substack{0\le\ell\le k^d\\\gcd(\ell,k)=1}}
 \binom{k^d}{\ell}(q-1)^\ell.
\end{equation*}
\end{theorem}

\begin{proof}
We establish the CPC property by providing an efficient decoder. Suppose $\bm{Z}=T_{\bm{a}}\bm{X}$ for some $\bm{X}\in\cB^{\mathrm{ms}}_{d,k}(\bm{\sigma})$ and $\bm{a}=(a_1,\ldots,a_d)\in\Z_k^d$. Translation preserves the Hamming weight and transforms the moments as
\begin{equation*}
S_j(T_{\bm{a}}\bm{X})
\equiv
S_j(\bm{X})-a_j\wt(\bm{X})
\pmod{k}, \qquad 1\le j \le d,
\end{equation*}
Since the weight is invertible modulo $k$, the inverse translation is uniquely determined coordinatewise by
\[
-a_j
\equiv
\bigl(S_j(\bm{Z})-\sigma_j\bigr)\wt(\bm{Z})^{-1}
\pmod{k},
\qquad 1\le j\le d.
\]
Applying $T_{-\bm{a}}$ to $\bm{Z}$ recovers $\bm{X}$, so $\cB^{\mathrm{ms}}_{d,k}(\bm{\sigma})$ is indeed a $d$D-CPC.

For the code size, let
\[
\cS
=
\left\{
\bm{X}\in\Sigma_q^{\Z_k^d}:
\gcd\bigl(\wt(\bm{X}),k\bigr)=1
\right\}.
\]
Since translation preserves Hamming weight, $\cS$ is invariant under translations. Fix $\bm{X}\in\cS$. For each $1\le j\le d$, as $a_j$ ranges over $\Z_k$, the value $S_j(\bm{X})-a_j\wt(\bm{X})$ runs through all residues modulo $k$ exactly once. Consequently, the map
\[
\bm{a}
\longmapsto
\bigl(
S_1(T_{\bm{a}}\bm{X}),
\ldots,
S_d(T_{\bm{a}}\bm{X})
\bigr)
\]
is a bijection from $\Z_k^d$ onto $\Z_k^d$. In particular, for every prescribed syndrome $\bm{\sigma}\in\Z_k^d$, there is a unique $\bm{a}\in\Z_k^d$ satisfying
\[
S_j(T_{\bm{a}}\bm{X})=\sigma_j,
\qquad
1\le j\le d.
\]
Thus each translation necklace within $\cS$ has size $k^d$ and contains exactly one representative from $\cB^{\mathrm{ms}}_{d,k}(\bm{\sigma})$. It follows that
\[
\left|\cB^{\mathrm{ms}}_{d,k}(\bm{\sigma})\right|
=
\frac{|\cS|}{k^d}
=
\frac{1}{k^d}
\sum_{\substack{0\le\ell\le k^d\\\gcd(\ell,k)=1}}
\binom{k^d}{\ell}(q-1)^\ell.
\]
This completes the proof.
\end{proof}

For prime $k$, the redundancy is asymptotically optimal.

\begin{lemma}\label{lem:MSC-d}
For fixed integers $q\ge2$ and $d\ge2$, as $k\to\infty$ through the
primes, the $q$-ary redundancy of
$\cB^{\mathrm{ms}}_{d,k}(\bm{\sigma})$ is
$\rho\big(\cB^{\mathrm{ms}}_{d,k}(\bm{\sigma})\big)=
d\log_q k+o(1)$.
\end{lemma}

\begin{proof}
Since $k$ is prime, $\gcd(\ell,k)=1$ if and only if $k\nmid\ell$.
Let
\[
S_k
=
\sum_{\substack{0\le \ell\le k^d\\ k\mid\ell}}
\binom{k^d}{\ell}(q-1)^\ell.
\]
By Theorem~\ref{thm:moment}, $\big|\cB^{\mathrm{ms}}_{d,k}(\bm{\sigma})\big|
=
\frac{q^{k^d}-S_k}{k^d}$. To estimate $S_k$, apply Lemma~\ref{lem:roots-unity-filter} to
$F(z)=(1+(q-1)z)^{k^d}$. With $\omega=e^{2\pi i/k}$ (here $i$ denotes the imaginary unit), we obtain
\[
S_k
=
\frac{q^{k^d}}{k}
+
\frac{1}{k}
\sum_{j=1}^{k-1}
\bigl(1+(q-1)\omega^j\bigr)^{k^d}.
\]

For $1\le j\le k-1$, let $m_j=\min\{j,k-j\}$.
By \eqref{eq:filter},
\[
\frac{\left|1+(q-1)\omega^j\right|^2}{q^2}
\le
1-\frac{16(q-1)m_j^2}{q^2k^2}.
\]
Using $1-x\le e^{-x}$, we obtain
\begin{align*}
\left|
\frac{1+(q-1)\omega^j}{q}
\right|^{k^d}
&\le
\left(
1-\frac{16(q-1)m_j^2}{q^2k^2}
\right)^{k^d/2}
\le
\exp\!\left(
-\frac{8(q-1)k^{d-2}}{q^2}m_j^2
\right).
\end{align*}
Since each positive
integer $m$ occurs as $m_j$ for at most two values of $j$, we have
\begin{align*}
\sum_{j=1}^{k-1}
\left|
\frac{1+(q-1)\omega^j}{q}
\right|^{k^d}
&\le
2\sum_{m=1}^{\infty}
\exp\!\left(
-\frac{8(q-1)k^{d-2}}{q^2}m^2
\right)
=
O(1),
\end{align*}
where the last equality holds since the series is a convergent Gaussian series.
By the triangle inequality,
\[
\left|
S_k-\frac{q^{k^d}}{k}
\right|
\le
\frac{q^{k^d}}{k}
\sum_{j=1}^{k-1}
\left|
\frac{1+(q-1)\omega^j}{q}
\right|^{k^d}
=
O\!\left(\frac{q^{k^d}}{k}\right).
\]
In particular,
$S_k=O_q(q^{k^d}/k)$ and
$\big|\cB^{\mathrm{ms}}_{d,k}(\bm{\sigma})\big|
=\big(1-O(1/k)\big)
q^{k^d}/k^d$.
It follows that
\begin{align*}
\rho\!\left(\cB^{\mathrm{ms}}_{d,k}(\bm{\sigma})\right)
&=
d\log_q k-\log_q\big(1-O(1/k)\big)
=
d\log_q k+o(1),
\end{align*}
which completes the proof.
\end{proof}

\section{Constructions for Error-Correcting CPCs}\label{sec:noisy}

We now extend the noiseless constructions from the previous section to the setting where an unknown translated array may be corrupted by at most $t$ substitutions.

\subsection{Robust Marker-Hyperplane Based Constructions}

The marker-hyperplane construction extends to the noisy setting by
requiring every internal coordinate hyperplane to have weight at
least $2t+1$. After at most $t$ substitutions, its weight remains at
least $t+1$, while the weight of an initially all-zero marker
hyperplane is at most $t$. Hence, the marker hyperplanes can still be
identified by their smaller weights.

For a coordinate direction $1\leq j\leq d$ and an internal index
$a\in[k]^+$, let 
\begin{equation*}
H_{j,a}
=
\left\{
\bm{i}\in([k]^+)^d:i_j=a
\right\}
\end{equation*}
be the corresponding hyperplane in $([k]^+)^d$. Define
\[
\cG_{d,k-1,t}
=
\left\{\bm U\in\Sigma_q^{([k]^+)^d}:
\wt(\bm U|_{H_{j,a}})\geq2t+1
\text{ for every }j,a\right\}.
\]
For $\bm{U}\in\cG_{d,k-1,t}$, define
$\bm{B}(\bm{U})\in\Sigma_q^{\Z_k^d}$ by placing $\bm{U}$ on
$([k]^+)^d$ and setting all remaining entries to $0$.
For $\cE\subseteq\cG_{d,k-1,t}$, define
\begin{equation*}
\cB^{\mathrm{mk}}_{d,k,t}(\cE)
=
\left\{
\bm{B}(\bm{U}):\bm{U}\in\cE
\right\}.
\end{equation*}

\begin{theorem}\label{thm:robust-marker}
Let $\cE\subseteq\cG_{d,k-1,t}$ have minimum Hamming distance at least
$2t+1$. Then $\cB^{\mathrm{mk}}_{d,k,t}(\cE)$ is a
$t$-substitution-correcting $d$D-CPC of size $|\cE|$. 
\end{theorem}

\begin{proof}
We establish the CPC property by describing a decoder. Suppose $\bm{Z}$ is received and $d_H(\bm{Z},T_{\bm{a}}\bm{B}(\bm{U}))\le t$ for some $\bm{U}\in\cE$ and $\bm{a}=(a_1,\ldots,a_d)\in\Z_k^d$. Before corruption, the marker hyperplane at index $0$ in each coordinate direction has weight $0$, while every other coordinate hyperplane has weight at least $2t+1$. After at most $t$ substitutions, the former has weight at most $t$, and each of the latter has weight at least $t+1$. For each $j$, the unique coordinate-$j$ hyperplane of $\bm{Z}$ with weight at most $t$ occurs at index $i_j\equiv -a_j\pmod{k}$. Thus $\bm{i}=(i_1,\ldots,i_d)=-\bm{a}$ determines the inverse translation. Applying $T_{\bm{i}}$ gives $\bm{Z}'$ with $d_H(\bm{Z}',\bm{B}(\bm{U}))\le t$.
Restricting $\bm{Z}'$ to $([k]^+)^d$ gives an array within distance $t$ of $\bm{U}$. Since $\cE$ has minimum distance at least $2t+1$, its bounded-distance decoder recovers $\bm{U}$ uniquely. Therefore $\bm{Z}$ determines both the translation and the codeword, so $\cB^{\mathrm{mk}}_{d,k,t}(\cE)$ is a $t$-substitution-correcting $d$D-CPC.

The map $\bm{U}\mapsto\bm{B}(\bm{U})$ is injective, so $\big|\cB^{\mathrm{mk}}_{d,k,t}(\cE)\big|=|\cE|$. This completes the proof.
\end{proof}

\begin{remark}
When $t=0$, taking $\cE=\cG_{d,k-1}$ recovers the noiseless marker-hyperplane construction of Theorem~\ref{thm:d-marker-noiseless}.
\end{remark}

For an explicit encoder, we force $2t+1$ disjoint internal diagonals to one.
Each such diagonal meets every internal coordinate hyperplane once. For
$i\in[k]^+$ and $s\in\Z_{k-1}$, write
$i\oplus s=1+\big((i-1+s)\bmod(k-1)\big)$ and set
\[
D_s=\{(i,i\oplus s,\ldots,i\oplus s):i\in[k]^+\}.
\]
Assume $2t+1\leq k-1$ and fix $D_0,\ldots,D_{2t}$ to one. The remaining
internal positions form a payload of length
$n_d(k,t)=(k-1)^d-(2t+1)(k-1)$.

For a payload word $\bm u$, let $\bm B^{\mathrm{fd}}(\bm u)$ be the array
with zero marker hyperplanes, ones on the selected diagonals, and $\bm u$ in
the remaining internal positions. For a code
$\cE'\subseteq\Sigma_q^{n_d(k,t)}$, set
$\cB^{\mathrm{fd}}_{d,k,t}(\cE')=
\{\bm B^{\mathrm{fd}}(\bm u):\bm u\in\cE'\}$.

\begin{lemma}\label{lemma:fixed-diagonal}
If $\cE'$ has minimum distance at least $2t+1$, then
$\cB^{\mathrm{fd}}_{d,k,t}(\cE')$ is a
$t$-substitution-correcting $d$D-CPC of redundancy
\begin{align*}
\rho\!\left(\mathcal B^{\mathrm{fd}}_{d,k,t}(\mathcal E')\right)
&=k^d-(k-1)^d
+
(2t+1)(k-1)
+
n_d(k,t)-\log_q|\mathcal E'|.
\end{align*}
\end{lemma}

\begin{proof}
Every internal hyperplane contains $2t+1$ fixed ones, so the internal arrays
belong to $\cG_{d,k-1,t}$. Distances between them equal the distances between
their payloads. By Theorem~\ref{thm:robust-marker}, $\cB^{\mathrm{fd}}_{d,k,t}(\cE')$ is a $t$-substitution-correcting $d$D-CPC. By construction, we have $\big| \mathcal B^{\mathrm{fd}}_{d,k,t}(\mathcal E')\big|= |\cE'|$. Then the redundancy follows by definition.
\end{proof}

Apart from the encoder and decoder of $\cE'$, filling the fixed
positions and finding the marker hyperplanes require $O(k^d)$ time.
It therefore remains to construct a $q$-ary code of length
$n_d(k,t)$ and minimum distance at least $2t+1$. We use a BCH code
over a larger finite field, whose proof is deferred to the appendix.

\begin{lemma}\label{lem:general-q-bch}
Let $q\geq2$ and $t\geq1$ be fixed, and let $Q\geq q$ be a prime
power. For all sufficiently large $n$, there exists a systematic
code $\cC\subseteq\Sigma_q^n$ with minimum Hamming distance at least
$2t+1$ and
\begin{equation}\label{eq:general-q-bch-redundancy}
\rho(\cC)
\leq
\left\lceil\log_q Q\right\rceil
\left(2t-\left\lfloor\frac{2t}{Q}\right\rfloor\right)
\left\lceil\log_Q(n+1)\right\rceil .
\end{equation}
For fixed $q$, $Q$, and $t$, encoding takes $O(n\log n)$ time and
decoding takes $O(n\log^2 n)$ time.
\end{lemma}

We now apply Lemma~\ref{lem:general-q-bch} with
$n=n_d(k,t)$. Fix a prime power $Q\geq q$. For all sufficiently
large $k$, let $\cE'$ be the code given by the lemma. Then
$\cE'$ has minimum distance at least $2t+1$ and redundancy at most
$$
\left\lceil\log_q Q\right\rceil
\left(2t-\left\lfloor\frac{2t}{Q}\right\rfloor\right)
\left\lceil
\log_Q\bigl(n_d(k,t)+1\bigr)
\right\rceil .
$$
Combining this with Lemma~\ref{lemma:fixed-diagonal} gives the
following result.

\begin{corollary}\label{cor:marker}
For fixed $d\geq2$, $q\geq2$, and $t\geq1$, fix a prime power
$Q\geq q$.
For all sufficiently large $k$, there exists an explicit
$t$-substitution-correcting $d$D-CPC with redundancy at most
\begin{align*}
k^d-(k-1)^d+(2t+1)(k-1)
+
\left\lceil\log_q Q\right\rceil
\left(2t-\left\lfloor\frac{2t}{Q}\right\rfloor\right)
\left\lceil
\log_Q\bigl(n_d(k,t)+1\bigr)
\right\rceil
=
O(k^{d-1}).
\end{align*}
Encoding takes $O(k^d\log k)$ time and decoding takes
$O(k^d\log^2 k)$ time.
\end{corollary}

\subsection{Robust Row-Anchor Based Construction}
\label{subsec:robust-row-anchor}

The row-anchor construction extends to the noisy setting by replacing the unique marker with one whose Hamming distance from every other row window is at least $2t+1$. After at most $t$ substitutions, the marker still determines the translation uniquely, and a global error-correcting code then recovers the entire array. We first give the general construction and then provide an explicit encoder and decoder.

Let $D=2t+1$ and fix a marker $\bm{\mu}\in\Sigma_q^m$ with $D<m\le k$.  For $\bm{X}\in\Sigma_q^{\Z_k^d}$, $\bm{i}\in\Z_k^{d-1}$, and $c\in\Z_k$, define
\begin{equation*}
\operatorname{win}_{\bm{i},c}(\bm{X})
=
\left(
X_{(\bm{i},(c+j)\bmod k)}
\right)_{j=0}^{m-1}.
\end{equation*}
Since $m\le k$, each coordinate occurs at most once in a window. By the definition of translation,
\[
\operatorname{win}_{\bm{i},c}\bigl(T_{(\bm{a},b)}\bm{X}\bigr)
=
\operatorname{win}_{\bm{i}+\bm{a},\,c+b}(\bm{X}).
\]

Let $\cC\subseteq\Sigma_q^{\Z_k^d}$ have minimum Hamming distance at least $D$. We require the marker at position $(\bm{0},0)$ and exclude every other window within distance $D-1$ of it. Define
\begin{equation*}
\cB^{\mathrm{row}}_{d,k,t}(\bm{\mu},\cC)
=
\left\{
\bm{X}\in\cC\;\middle|\;
\begin{array}{l}
\operatorname{win}_{\bm{0},0}(\bm{X})=\bm{\mu},\\[2pt]
d_H\bigl(\operatorname{win}_{\bm{i},c}(\bm{X}),\bm{\mu}\bigr)\geq D
\text{ for every }(\bm{i},c)\neq(\bm{0},0)
\end{array}
\right\}.
\end{equation*}

\begin{theorem}\label{thm:robust-row-anchor}
Let $\cC\subseteq\Sigma_q^{\Z_k^d}$ have minimum Hamming distance at least $D$. The code $\cB^{\mathrm{row}}_{d,k,t}(\bm{\mu},\cC)$ is a $t$-substitution-correcting $d$D-CPC.
\end{theorem}

\begin{proof}
We establish the CPC property by describing a decoder. Suppose $d_H\bigl(\bm{Z},T_{(\bm{a},b)}\bm{X}\bigr)\le t$ for some $\bm{X}\in\cB^{\mathrm{row}}_{d,k,t}(\bm{\mu},\cC)$. Before corruption, the marker occurs at $(-\bm{a},-b)$ in the translated array, so
\[
d_H\bigl(\operatorname{win}_{-\bm{a},-b}(\bm{Z}),\bm{\mu}\bigr)\le t.
\]
For any other pair $(\bm{i},c)$, the corresponding window before corruption is a non-marker window of $\bm{X}$ and therefore has distance at least $D$ from $\bm{\mu}$. Since the whole array contains at most $t$ substitutions, its distance from $\bm{\mu}$ after corruption is at least $D-t=t+1$. Thus $(-\bm{a},-b)$ is the unique pair whose window lies within distance $t$ of $\bm{\mu}$, and the translation is uniquely determined.

After applying the inverse translation, the received array is within distance $t$ of $\bm{X}$. Since $d_H(\cC)\ge D=2t+1$, the codeword $\bm{X}$ is also uniquely determined. Hence the received array uniquely determines both the codeword and the translation. This completes the proof.
\end{proof}

We next give an explicit encoder and decoder. Fix $q\ge2$, $d\ge2$, and $t\ge1$, and set $D=2t+1$.

\begin{definition}\label{def:row-anchor-wwl}
Let $1\le\Delta\le\ell\le n$. A sequence $\bm v\in\Sigma_q^n$ is
\emph{$(\ell,\Delta)$-window-weight-limited}, or \emph{$(\ell,\Delta)$-WWL}, if every non-cyclic length-$\ell$ window has Hamming weight at least $\Delta$. It is \emph{cyclically $(\ell,\Delta)$-WWL} if this holds for every cyclic window. An array is \emph{row-cyclic $(\ell,\Delta)$-WWL} if each of its rows is cyclically $(\ell,\Delta)$-WWL.
\end{definition}

\begin{remark}\label{rmk:wwl-one-symbol-extension}
If $\bm v\in\Sigma_q^n$ is cyclically $(\ell,\Delta+1)$-WWL, then $\bm v\circ\eta$ is cyclically $(\ell,\Delta)$-WWL for every $\eta\in\Sigma_q$.
\end{remark}

The following WWL coding result is proved in Appendix~\ref{appendixB}.

\begin{lemma}\label{lem:row-cyclic-wwl}
Let $q\ge2$ and let $n$ be a positive multiple of $\kappa$. Suppose $1\le\Delta\le\ell$, $2\ell\le\kappa$, and
\begin{equation}\label{eq:WWL-pointer-capacity}
    q^{\ell-\Delta}\ge nV_q(\ell,\Delta-1).
\end{equation}
Then there exists an injective encoder
$\mathrm{ENC}^{\mathrm{WWL}}_{n,\kappa,\Delta}[\ell]:
    \Sigma_q^{n-1}\longrightarrow
    \Sigma_q^{(n/\kappa)\times\kappa}$
whose output is row-cyclic $(\ell,\Delta)$-WWL. Denote its inverse by $\mathrm{DEC}^{\mathrm{WWL}}_{n,\kappa,\Delta}[\ell]$. For fixed $q$ and $\Delta$, both algorithms take $O(n\ell^2)$ time.
\end{lemma}

To distinguish overlapping marker positions, we use the auto-cyclic sequence of Levy and Yaakobi \cite[Section~VI]{LevyYaakobi2019}.

\begin{definition}\label{def:row-anchor-auto-cyclic}
A sequence $\bm u$ of length $\lambda\ge D$ is \emph{$D$-auto-cyclic} if
\begin{equation}\label{eq:anchor-auto-cyclic-property}
    d_H\bigl(\bm u,(0^i\circ\bm u)_{[0,\lambda-1]}\bigr)\ge D,
    \qquad 1\le i\le D.
\end{equation}
\end{definition}

\begin{lemma}{\cite[Section~VI]{LevyYaakobi2019}}\label{lem:row-anchor-auto-cyclic}
For $0\le j<\lceil\log_2D\rceil$, set $\bm u_j=\bigl((1^{2^j}\circ0^{2^j})^D\bigr)_{[0,D-1]}$. Then
\[
    \bm u=1^D\circ\bm u_0\circ\cdots\circ
    \bm u_{\lceil\log_2D\rceil-1}
\]
is $D$-auto-cyclic and has length $\lambda(D)=D\lceil\log_2D\rceil+D$.
\end{lemma}

Set $R=k^{d-1}-1$ and $N=(k-1)R$, and fix an ordering $\bm i_1,\ldots,\bm i_R$ of $\Z_k^{d-1}\setminus\{\bm0\}$. Fix a prime power $Q\ge q$. By Lemma~\ref{lem:general-q-bch}, choose a systematic code $\cC_{\mathrm{BCH}}\subseteq\Sigma_q^{k^d}$ with minimum distance at least $D$ and integer redundancy
\[
    r_{\mathrm{BCH}}\le
    \left\lceil\log_q Q\right\rceil
    \left(2t-\left\lfloor\frac{2t}{Q}\right\rfloor\right)
    \left\lceil\log_Q(k^d+1)\right\rceil.
\]
Write its encoder as $\mathrm{ENC}_{\mathrm{BCH}}(\bm s)=\bm s\circ\bm\gamma(\bm s)$, where $\bm\gamma(\bm s)\in\Sigma_q^{r_{\mathrm{BCH}}}$.

Let 
\[
    \ell=\min\{L\geq D+1: q^{L-D-1}\geq NV_q(L,D)\},
\]
we have
\begin{equation}\label{eq:marker-length}
\begin{aligned}
    \ell&=\log_q N+D\log_q\log_q N+O(1)=d\log_q k+D\log_q\log_q k+O(1).
\end{aligned}
\end{equation}
Set $n_{\mathrm A}=k-\ell-\lambda(D)-2D$ and assume
\[
    r_{\mathrm{BCH}}<R,\qquad 2\ell\le k-1,\qquad
    \ell\ge\lambda(D)+2D,\qquad n_{\mathrm A}\ge2\ell.
\]
These conditions hold for sufficiently large $k$. The message consists of
\[
    \bm x_{\mathrm A}\in\Sigma_q^{n_{\mathrm A}-1},\qquad
    \bm x_{\mathrm W}\in\Sigma_q^{N-1},\qquad
    \bm z\in\Sigma_q^{R-r_{\mathrm{BCH}}}.
\]

\subsubsection{Anchor-row encoding} Let $\bm u$ be the $D$-auto-cyclic sequence defined in Lemma~\ref{lem:row-anchor-auto-cyclic}. We define the synchronization marker as
\begin{equation*}
    \bm\mu=1^D\circ0^\ell\circ\bm u\circ1^D.
\end{equation*}
By the choice of $\ell$ in \eqref{eq:marker-length}, we have $q^{\ell-D}\ge qNV_q(\ell,D)\ge n_{\mathrm A}V_q(\ell,D-1)$.
Thus we can use the WWL encoder $\mathrm{ENC}^{\mathrm{WWL}}_{n_{\mathrm A},n_{\mathrm A},D}[\ell]$ from Lemma~\ref{lem:row-cyclic-wwl} to encode $\bm{x}_{\mathrm A}\in\Sigma_q^{n_{\mathrm A}-1}$ into a cyclically $(\ell,D)$-WWL sequence of length $n_{\mathrm A}$. Finally, the anchor row is encoded as
\[
    \bm a_{\mathrm A}
    =\bm\mu\circ \bm{c}
    \in\Sigma_q^k, \qquad \text{where} \qquad \bm{c}=\mathrm{ENC}^{\mathrm{WWL}}_{n_{\mathrm A},n_{\mathrm A},D}[\ell](\bm{x}_{\mathrm A}).
\]

\subsubsection{Joint encoding of non-anchor row prefixes}
By the choice of $\ell$, we have $q^{\ell-D-1}\ge NV_q(\ell,D)$. Thus, we can use WWL encoder $\mathrm{ENC}^{\mathrm{WWL}}_{N,k-1,D+1}[\ell]$ to encode $\bm{x}_W\in \Sigma_q^{N-1}$ into a cyclically $(\ell,D+1)$ array of size $\frac{N}{k-1}\times (k-1)$.
Set
\begin{equation*}
    \bm y=\mathrm{ENC}^{\mathrm{WWL}}_{N,k-1,D+1}[\ell]
    (\bm x_{\mathrm W})=\bm y_1\circ\cdots\circ\bm y_R \in \Sigma_q^N,
\end{equation*}
where each $\bm y_j\in\Sigma_q^{k-1}$ is cyclically $(\ell,D+1)$-WWL. By Remark~\ref{rmk:wwl-one-symbol-extension}, appending any symbol to $\bm y_j$ gives a cyclically $(\ell,D)$-WWL row.

\subsubsection{Global error correction and placement}
Form the systematic vector
\begin{equation*}
    \bm s=\bm a_{\mathrm A}\circ\bm y\circ\bm z
    \in\Sigma_q^{k^d-r_{\mathrm{BCH}}},
\end{equation*}
and compute $\bm\gamma(\bm s)=(p_1,\ldots,p_{r_{\mathrm{BCH}}})$. Place these symbols as follows:
\begin{equation*}
\begin{aligned}
    \operatorname{row}_{\bm0}(\bm X)&=\bm a_{\mathrm A},\\
    \operatorname{row}_{\bm i_j}(\bm X)&=\bm y_j\circ p_j,
    &&1\le j\le r_{\mathrm{BCH}},\\
    \operatorname{row}_{\bm i_j}(\bm X)
    &=\bm y_j\circ z_{j-r_{\mathrm{BCH}}},
    &&r_{\mathrm{BCH}}<j\le R.
\end{aligned}
\end{equation*}
Let $\cB^{\mathrm{row,exp}}_{d,k,t}$ denote the encoder image, and let $\Pi$ be the fixed coordinate permutation implementing this placement. Then $\bm X=\Pi(\bm s\circ\bm\gamma(\bm s))$, so the BCH-based code simultaneously protects the anchor row, all row prefixes, and $\bm z$.
Set
\[
\widetilde{\cC}_{\mathrm{BCH}}=\Pi(\cC_{\mathrm{BCH}}).
\]

\subsubsection{Decoding}
Given $\bm Z$ within distance $t$ of some translate of $\bm X$, find the unique window starting at $(\bm i_0,c_0)$ that is within distance $t$ of $\bm\mu$. If no such window exists or it is not unique, declare failure. Apply $T_{(\bm i_0,c_0)}$, followed by $\Pi^{-1}$ and the BCH-based decoder, to recover $\bm s$. Extract $\bm c$ and apply $\mathrm{ENC}^{\mathrm{WWL}}_{n_{\mathrm A},n_{\mathrm A},D}[\ell]$ to recover $\bm x_{\mathrm A}$. Apply $\mathrm{DEC}^{\mathrm{WWL}}_{N,k-1,D+1}[\ell]$ to $\bm y$ to recover $\bm x_{\mathrm W}$, and read $\bm z$ directly.

\begin{theorem}\label{thm:explicit-robust-row-anchor}
For fixed $q\ge2$, $d\ge2$, $t\ge1$, and prime power $Q\ge q$, for sufficiently large $k$, 
\begin{equation*}
    \cB^{\mathrm{row,exp}}_{d,k,t}
    \subseteq\cB^{\mathrm{row}}_{d,k,t}
    (\bm\mu,\widetilde{\cC}_{\mathrm{BCH}})
\end{equation*}
is a $t$-substitution-correcting $d$D-CPC with redundancy
\begin{equation*}
\rho\bigl(\cB^{\mathrm{row,exp}}_{d,k,t}\bigr)
    =\ell+\lambda(D)+2D+2+r_{\mathrm{BCH}}
    \le d\log_q k+D\log_q\log_q k+r_{\mathrm{BCH}}+O_{q,d,t}(1).
\end{equation*}
Both encoding and decoding take $O(k^d\log^2 k)$ word operations for fixed $q,Q,d,t$.
\end{theorem}

\begin{proof}
Write $\lambda=\lambda(D)$. By construction, every output $\bm X$ lies in $\widetilde{\cC}_{\mathrm{BCH}}$, and $\Pi$ preserves Hamming distance, so $d_H(\widetilde{\cC}_{\mathrm{BCH}})\ge D$. The anchor row satisfies
\[
    \operatorname{win}_{\bm0,0}(\bm X)=\bm\mu=1^D\circ0^\ell\circ\bm u\circ1^D,
\]
where $\bm u$ is the $D$-auto-cyclic sequence beginning with $1^D$ and $D<|\bm\mu|=\ell+\lambda+2D\le k$. It remains to show that every other window $\bm V_s$ of length $|\bm\mu|$ starting at position $s\ne0$ satisfies $d_H(\bm V_s,\bm\mu)\ge D$.
We use two facts.
\begin{itemize}
    \item Since $\bm c$ is cyclically $(\ell,D)$-WWL, the word $1^D\circ\bm c\circ1^D$ is $(\ell,D)$-WWL.
    \item Since $\bm u$ is $D$-auto-cyclic,
    \[
        d_H\bigl(\bm u,\,0^i\circ\bm u_{[0,\lambda-i-1]}\bigr)\ge D,
        \qquad 1\le i\le D.
    \]
\end{itemize}

Let $\bm W_s$ be the length-$\ell$ subblock of $\bm V_s$ aligned with $0^\ell$. Since the corresponding coordinates of $\bm\mu$ are zero,
\[
    d_H(\bm V_s,\bm\mu)\ge d_H(\bm W_s,0^\ell)=\wt(\bm W_s).
\]
We distinguish two cases.

\emph{Case 1: $\bm W_s$ avoids the zero block.}
If $\bm W_s$ also avoids $\bm u$, then $\bm W_s$ is a window of $1^D\circ\bm c\circ1^D$, so $\wt(\bm W_s)\ge D$ by the first fact. If $\bm W_s$ meets $\bm u$, then it extends past $\bm u$ into the following $1^D$, so again $\wt(\bm W_s)\ge D$.

\emph{Case 2: $\bm W_s$ intersects the zero block.}
Then $s\equiv i$ or $s\equiv-i\pmod k$ for a unique $i\in\{1,\dots,\ell-1\}$. If $i\ge D$, then $\bm W_s$ contains a block $1^D$ (either the prefix of $\bm a_{\mathrm A}$ or the prefix of $\bm u$), so $\wt(\bm W_s)\ge D$.
If $1\le i<D$, set
\[
    \Delta_i=d_H\bigl(\bm u_{[i,\lambda-1]},\bm u_{[0,\lambda-i-1]}\bigr).
\]
By the auto-cyclic property of $\bm u$,
\[
    D\le d_H\bigl(\bm u,\,0^i\circ\bm u_{[0,\lambda-i-1]}\bigr)=i+\Delta_i.
\]
We consider the two possible positions.
\begin{itemize}
    \item If $s\equiv-i\pmod k$, the length-$\lambda$ window of $\bm V_s$ aligned with $\bm u$ equals $0^i\circ\bm u_{[0,\lambda-i-1]}$, giving at least $D$ mismatches.
    \item If $s\equiv i\pmod k$, then $\bm V_s=1^{D-i}\circ0^\ell\circ\bm u\circ1^D\circ\bm c_{[0,i-1]}$.
    We can compute
    \begin{align*}
        d_H(\bm V_s,\bm\mu)
        &\ge d_H\bigl(1^{D-i}\circ0^\ell\circ\bm u,\,
        1^D\circ0^\ell\circ\bm u_{[0,\lambda-i-1]}\bigr)\\
        &\ge d_H\bigl(\bm u,\,0^i\circ\bm u_{[0,\lambda-i-1]}\bigr)\ge D.
    \end{align*}
\end{itemize}

Therefore,
\[
    d_H\bigl(\operatorname{win}_{\bm i,c}(\bm X),\bm\mu\bigr)\ge D
    \qquad\text{for every }(\bm i,c)\ne(\bm0,0).
\]
This proves that $\cB^{\mathrm{row,exp}}_{d,k,t}
\subseteq\cB^{\mathrm{row}}_{d,k,t}
(\bm\mu,\widetilde{\cC}_{\mathrm{BCH}})$ is a $t$-substitution-correcting $d$D-CPC.
Since the information length is
\[
    K=(n_{\mathrm A}-1)+(N-1)+(R-r_{\mathrm{BCH}}),
\]
the redundancy bound follows.

We finally bound the running time. The encoder uses the BCH encoder and the WWL encoder, which together take $O(k^d\ell^2)$ time. The decoder uses the BCH decoder and the WWL decoder, and additionally checks the Hamming distance between the marker $\bm\mu$ and each window of the same length, which also takes $O(k^d\ell^2)$ time. This completes the proof.
\end{proof}

For $q=Q=2$, we have $r_{\mathrm{BCH}}\le td\log_2 k+O(1)$, and hence
\[
    \rho\bigl(\cB^{\mathrm{row,exp}}_{d,k,t}\bigr)
    \le(t+1)d\log_2 k+(2t+1)\log_2\log_2 k+O(1).
\]
By Corollary~\ref{cor:redundancy-noisy-CPC}, this matches the optimal leading term.

\begin{remark}
    For the code construction, one may also take the anchor row from a $t$-substitution-correcting one-dimensional CPC, require every non-anchor row to have distance at least $2t+1$ from the translation closure of the anchor code, and require the concatenation of the non-anchor rows to belong to a code of minimum distance at least $2t+1$. The anchor row then determines the translation, while the code on the remaining rows corrects the errors. When $t=0$, this reduces to the noiseless row-anchor construction of Theorem~\ref{thm:row-anchor-multi}. By comparison, the framework described in this subsection is more convenient for the explicit construction, because a single systematic error-correcting code protects the entire array and its parity symbols can be inserted without violating the non-anchor row constraint.
\end{remark}

\subsection{Robust Moment-Syndrome Based Construction}
\label{subsec:robust-moment}

We now turn to the moment-syndrome approach. In the noisy setting, errors may perturb both the moments and the Hamming weight of the received array. We therefore first correct the errors using a translation-invariant inner code, and then recover the translation from the moments of the corrected array.

\begin{definition}
\label{def:translation-invariant-code}
A code $\cD\subseteq\Sigma_q^{\Z_k^d}$ is \emph{translation invariant} if $T_{\bm a}\bm X\in\cD$ for every $\bm X\in\cD$ and every $\bm a\in\Z_k^d$.
\end{definition}

Translation invariance ensures that after an unknown translation, the transmitted array remains in the same inner code.
For a translation-invariant code
$\cD\subseteq\Sigma_q^{\Z_k^d}$,
let $A_\ell(\cD)$ be the number of codewords in $\cD$ of Hamming weight
$\ell$. For a fixed syndrome
$\bm\sigma=(\sigma_1,\ldots,\sigma_d)$, define
\[
\cB^{\mathrm{ms}}_{d,k,t}(\bm\sigma;\cD)
=
\left\{
\bm X\in\cD:
\gcd\bigl(\wt(\bm X),k\bigr)=1,\;
S_j(\bm X)=\sigma_j
\text{ for }1\le j\le d
\right\}.
\]

\begin{theorem}
\label{thm:moment-error-correcting}
Let $\cD\subseteq\Sigma_q^{\Z_k^d}$ be a translation-invariant code with
minimum Hamming distance at least $2t+1$. For every
$\bm\sigma\in\Z_k^d$,
$\cB^{\mathrm{ms}}_{d,k,t}(\bm\sigma;\cD)$ is a
$t$-substitution-correcting $d$D-CPC of size
\begin{equation}
\left|
\cB^{\mathrm{ms}}_{d,k,t}(\bm\sigma;\cD)
\right|
=
\frac{1}{k^d}
\sum_{\substack{0\le \ell\le k^d\\
\gcd(\ell,k)=1}}
A_\ell(\cD).
\label{eq:robust-moment-size}
\end{equation}
\end{theorem}

\begin{proof}
We establish the CPC property by providing a decoder.
Suppose $\bm{Z}$ is received and $d_H(\bm Z,T_{\bm a}\bm X)\le t$ for some
$\bm X\in\cB^{\mathrm{ms}}_{d,k,t}(\bm\sigma;\cD)$
and
$\bm a=(a_1,\ldots,a_d)\in\Z_k^d$.
Set $\bm Y=T_{\bm a}\bm X$.
Since $\cD$ is translation invariant and $d_H(\cD)\geq 2t+1$, $\bm Y\in\cD$ is the unique codeword within distance $t$ of $\bm Z$. Thus, it can be recovered by a bounded-distance decoder for $\cD$.
It remains to recover the translation. Translation preserves Hamming weight and changes the moments as
\[
S_j(\bm Y)
\equiv
S_j(\bm X)-a_j\wt(\bm X)
\pmod{k},
\qquad
1\le j\le d.
\]
Since the weight is invertible modulo $k$, the inverse translation is determined coordinatewise by
\begin{equation}
-a_j
\equiv
\bigl(S_j(\bm Y)-\sigma_j\bigr)
\wt(\bm Y)^{-1}
\pmod{k},
\qquad
1\le j\le d.
\label{eq:robust-moment-translation}
\end{equation}
Applying $T_{-\bm a}$ to $\bm Y$ recovers $\bm X$.
Thus $\bm Z$ determines both the codeword and the translation, so
$\cB^{\mathrm{ms}}_{d,k,t}(\bm\sigma;\cD)$ is a
$t$-substitution-correcting $d$D-CPC.

For the code size, define
\[
\cD^\ast
=
\left\{
\bm X\in\cD:
\gcd\bigl(\wt(\bm X),k\bigr)=1
\right\}.
\]
Both $\cD$ and Hamming weight are translation invariant, so $\cD^\ast$
is also translation invariant.
For $\bm X\in\cD^\ast$,
\[
S_j(T_{\bm a}\bm X)
\equiv
S_j(\bm X)-a_j\wt(\bm X)
\pmod{k},
\qquad
1\le j\le d.
\]
Since $\wt(\bm X)$ is invertible modulo $k$, the map
\[
\bm a
\longmapsto
\bigl(
S_1(T_{\bm a}\bm X),\ldots,
S_d(T_{\bm a}\bm X)
\bigr)
\]
is a bijection from $\Z_k^d$ onto $\Z_k^d$.
In particular, for every
$\bm\sigma\in\Z_k^d$, there is a unique
$\bm a\in\Z_k^d$ with
\[
S_j(T_{\bm a}\bm X)=\sigma_j,
\qquad
1\le j\le d.
\]
Thus each translation necklace within $\cD^\ast$ has size
$k^d$ and contains exactly one array with moment syndrome $\bm\sigma$.
It follows that
\[
\left|
\cB^{\mathrm{ms}}_{d,k,t}(\bm\sigma;\cD)
\right|
=
\frac{|\cD^\ast|}{k^d}
=\frac{1}{k^d}
\sum_{\substack{0\le \ell\le k^d\\
\gcd(\ell,k)=1}}
A_\ell(\cD).
\]
This completes the proof.
\end{proof}

\begin{remark}
Determining the size of the code in Theorem~\ref{thm:moment-error-correcting} requires the weight distribution $A_\ell(\cD)$, which is generally difficult to obtain for arbitrary translation-invariant codes. For linear translation-invariant codes, i.e., multidimensional cyclic codes, the algebraic structure gives some control over the dimension and minimum distance (see, e.g., \cite{GuneriOzbudak2008,BernalBuenoSimon2016,BernalGuerreiroSimon2019}). However, explicit weight distributions are known only in special cases. Determining $A_\ell(\cD)$ for general multidimensional cyclic codes, and hence obtaining explicit redundancy estimates for the resulting robust moment-syndrome CPCs, remains an open problem for future work.
\end{remark}

\section{Applications to Single-Fragment Forensic Coding}
\label{sec:forensic}

In this section, we apply the CPCs developed above to single-fragment
forensic coding. By Theorem~\ref{thm:periodic-lifting}, a
$t$-substitution-correcting $d$D-CPC of side length $k$ gives a
$t$-substitution-correcting $(M,h)$ forensic code whenever $k\leq h$.
Throughout this section, $d\geq2$ and the ambient side lengths satisfy
$n_j\geq h$ for $1\leq j\leq d$ and
$\prod_{j=1}^d n_j\geq M$, so legal fragments exist.

Following \cite{LiuRavivISIT2025,LiuRaviv2025}, we normalize the rate by the guaranteed fragment volume, namely,
\[
R_{\mathrm{SF}}(\mathcal C)
=
\frac{\log_q|\mathcal C|}{M}.
\]
That is, the rate measures the amount of information per symbol guaranteed to be observed by the decoder, rather than per symbol of the ambient object.
We consider the parameter regime $h=c M^{1/d}$ with $0<c\leq 1$ and take $k=h$ when this is an integer. When rounding is necessary, taking $k=\lfloor h\rfloor$ alters the rates by $O(M^{-1/d})$. For a side-$k$ $d$D-CPC of redundancy $\rho$, the corresponding $(M,h)$-single-fragment forensic code has rate
\begin{equation}\label{eq:rate}
R_{\mathrm{SF}}
=
\frac{k^d-\rho}{M}
=
c^d-\frac{\rho}{M}.
\end{equation}

\subsection{Noiseless Single-Fragment Forensic Coding}

We now apply the noiseless CPC constructions developed in Section~\ref{sec:noiseless_multi} to single-fragment forensic coding. The following theorem summarizes the resulting parameters.

\begin{theorem}
Fix $d\geq2$, $q\geq2$, and $0<c\leq1$. Let
$k=h=cM^{1/d}$ be a sufficiently large integer. Periodic lifting gives
the following noiseless $(M,h)$ single-fragment forensic codes over $\Sigma_q^{[n_1]\times\cdots\times[n_d]}$.
\begin{enumerate}
\item The explicit marker-hyperplane construction has rate
$c^d-O(M^{-1/d})$, with encoding/decoding time $O(M+\prod_{i=1}^d n_i)$.

\item The explicit row-anchor construction has rate
$c^d-\big(\log_q M+O(1)\big)/M$, with encoding/decoding time $O(M\log^2 M+\prod_{i=1}^d n_i)$.

\item For prime $k$, the moment-syndrome construction has rate
$c^d-\big(\log_q M+O(1)\big)/M$ and admits a direct synchronization decoder.
\end{enumerate}
\end{theorem}

\begin{proof}
The conclusion follows by applying Theorem~\ref{thm:periodic-lifting} and substituting the redundancy of each CPC construction into Equation~\eqref{eq:rate}.
\end{proof}

\subsection{Substitution-Correcting Single-Fragment Forensic Coding}

We next turn to the noisy setting. Among the substitution-correcting CPC constructions developed in Section~\ref{sec:noisy}, the marker-diagonal and row-anchor constructions admit explicit redundancy bounds. Applying them yields the following result.

\begin{theorem}
Fix $d\geq2$, $q\geq2$, $t\geq1$, $0<c\leq1$, and a prime
power $Q\geq q$. Let $k=h=cM^{1/d}$ be a sufficiently large
integer. Periodic lifting gives the following
$t$-substitution-correcting $(M,h)$ single-fragment forensic codes over $\Sigma_q^{[n_1]\times\cdots\times[n_d]}$.
\begin{enumerate}
\item The explicit robust marker-hyperplane construction has rate $c^d-O(M^{-1/d})$, with encoding time $O(M \log M + \prod_{i=1}^d n_i)$ and decoding time $O(M \log^2 M + \prod_{i=1}^d n_i)$.

\item The explicit robust row-anchor construction has rate
$c^d-O(\log_q M)/M$, with encoding/decoding time $O(M \log^2 M + \prod_{i=1}^d n_i)$.
\end{enumerate}
\end{theorem}

\begin{proof}
The conclusion follows by applying Theorem~\ref{thm:periodic-lifting} and substituting the redundancy of each CPC construction into Equation~\eqref{eq:rate}.
\end{proof}

\section{Conclusion and Open Problems}\label{sec:conclusion}

We studied multidimensional CPCs and used them as base codes for single-fragment forensic coding. Periodic lifting separates the problem into two parts, where repetition spreads information through the object and the CPC recovers an unknown cyclic translation of one complete period.

We established the optimal redundancy $d\log_q k+o(1)$ for noiseless $d$-dimensional CPCs, and developed three constructions. The marker-hyperplane family has redundancy $k^d-(k-1)^d+o(1)$, with a linear-time explicit subcode of redundancy $k^d-(k-1)^d+k-1$. The row-anchor family has redundancy $d\log_q k+\log_q e+o(1)$, while its explicit subcode has redundancy
$\lceil\log_q(k^d-k)\rceil+5$ and encoding and decoding time
$O(k^d\log^2 k)$. For prime $k$, the moment-syndrome construction has redundancy $d\log_q k+o(1)$, and admits a direct synchronization decoder. In the noisy setting, the explicit robust row-anchor code
achieves $(t+1)d\log_2 k+O(\log\log k)$ redundancy for $q=2$, matching the optimal leading term. For general alphabets, the proved bounds leave a gap between the coefficients $t+1$ and $2t+1$.
Periodic lifting gives asymptotic forensic rate $c^d$ in the thick-fragment
regime, with rate tending to one as $c\to1$.

Several questions remain open.

\begin{itemize}
\item \emph{Efficient encoding for the moment-syndrome construction.}
The moment-syndrome construction has nearly-optimal redundancy and an explicit decoder, but lacks an efficient encoder. Developing such an encoder is left for future work.

\item \emph{Weight distribution of translation-invariant codes.}
Evaluating the robust moment-syndrome construction in Theorem~\ref{thm:moment-error-correcting} requires the weight distribution of a translation-invariant code. Determining this distribution is left for future work.

\item \emph{Intermediate fragment thickness.}
The periodic lifting approach is most effective when $h=\Theta(M^{1/d})$, whereas the discrepancy-based constructions of \cite{LiuRaviv2025} handle thinner fragments but offer weaker rate guarantees. Hybrid constructions bridging these two regimes remain for future work.

\item \emph{Weakly cyclically permutable codes.}
The forensic problem does not require translation recovery. Relaxing the atranslational requirement may lead to larger codes or simpler constructions. A systematic study of weakly cyclically permutable codes is left for future work.
\end{itemize}

\appendices

\section{Proof of Lemma~\ref{lem:general-q-bch}}\label{appendixA}
Set $L=\lceil\log_q Q\rceil$ and $m=\lceil\log_Q(n+1)\rceil$. Consider a primitive narrow-sense BCH code (see \cite[Section~5.1]{HuffmanPless2003}) over $\F_Q$ of length $Q^m-1$ and minimum distance at least $2t+1$. Let $r$ denote its redundancy. Then
\[
r\leq
\left(2t-\left\lfloor\frac{2t}{Q}\right\rfloor\right)m.
\]
Since $r=O(\log n)$, we have $n-Lr\geq1$ for all sufficiently large $n$.

Put the BCH code in systematic form and shorten it in information positions until its dimension is $n-Lr$. The original dimension is $Q^m-1-r$, so the number of shortened positions is $Q^m-1-n+(L-1)r$. Hence the shortened code has dimension $n-Lr$, length $n-(L-1)r$, and redundancy $r$. Shortening does not decrease the minimum distance, so the minimum distance remains at least $2t+1$.

Since $Q\ge q$ and $q^L\ge Q$, we can fix injections $\iota:\Sigma_q\to\mathbb{F}_Q$ and $\phi:\mathbb{F}_Q\to\Sigma_q^L$, which extend coordinatewise to vectors. To encode a message $\boldsymbol{x}\in\Sigma_q^{n-Lr}$, first apply $\iota$ to obtain $\iota(\boldsymbol{x})\in\mathbb{F}_Q^{n-Lr}$, then encode $\iota(\boldsymbol{x})$ with the shortened systematic BCH code. This produces a codeword $\iota(\boldsymbol{x})\circ \gamma(\boldsymbol{x}) \in\mathbb{F}_Q^{n-(L-1)r}$, where $\gamma(\boldsymbol{x})\in\mathbb{F}_Q^r$ is the parity vector. The final $q$-ary codeword is the concatenation $\boldsymbol{x}\circ\phi\big(\gamma(\boldsymbol{x})\big)$. Its length is $(n-Lr)+Lr=n$, and the first $n-Lr$ symbols are unchanged from the original message $\boldsymbol{x}$. Thus the resulting $q$-ary code is systematic and has redundancy $Lr$.
For two distinct messages $\boldsymbol{x},\boldsymbol{x}'\in\Sigma_q^{n-Lr}$, we have
\begin{align*}
    d_H\big(\boldsymbol{x}\circ\phi(\gamma(\boldsymbol{x})), \boldsymbol{x}'\circ\phi(\gamma(\boldsymbol{x}'))\big)
    &= d_H(\boldsymbol{x}, \boldsymbol{x}')+d_H\big(\phi(\gamma(\boldsymbol{x})), \phi(\gamma(\boldsymbol{x}'))\big)\\
    &\geq d_H\big(\iota(\boldsymbol{x}), \iota(\boldsymbol{x}')\big)+d_H\big(\gamma(\boldsymbol{x}), \gamma(\boldsymbol{x}')\big)\\
    &= d_H\big(\iota(\boldsymbol{x}) \circ \gamma(\boldsymbol{x}), \iota(\boldsymbol{x}')\circ \gamma(\boldsymbol{x}')\big).
\end{align*}
Therefore, the minimum distance is at least $2t+1$.

For decoding, apply $\iota$ to the first $n-Lr$ received symbols. Then split the remaining redundancy symbols into $r$ blocks of length $L$. If a block belongs to $\phi(\F_Q)$, map it back with $\phi^{-1}$; otherwise, map it to any fixed element of $\F_Q$. A substitution in the $q$-ary word affects at most one symbol of the resulting $\F_Q$ word. Therefore, at most $t$ substitutions produce at most $t$ errors in the shortened BCH codeword. Its bounded-distance decoder recovers the transmitted BCH codeword and hence the original $q$-ary message.

Finally, the redundancy bound follows from the bound on $r$ above. For fixed $q$, $Q$, and $t$, we have $r=O(\log n)$. Systematic BCH encoding therefore takes $O(n\log n)$ operations, including the alphabet conversion. Standard bounded-distance BCH decoding takes $O(n\log^2 n)$ operations over $\F_Q$ \cite[Section~5.4]{HuffmanPless2003}. This completes the proof.

\section{Proof of Lemma~\ref{lem:row-cyclic-wwl}}\label{appendixB}
Let $\kappa\mid n$. We view a word $\bm y\in\Sigma_q^n$ as $n/\kappa$ rows of length $\kappa$. For $i=r\kappa+c$ with $0\le c<\kappa$, define the row-cyclic window
\[
\operatorname{rwin}^{(\ell)}_i(\bm y)
=\bigl(y_{r\kappa+((c+j)\bmod\kappa)}\bigr)_{j=0}^{\ell-1}.
\]
Set $B=V_q(\ell,\Delta-1)$ and order the forbidden words lexicographically. Let $\operatorname{rank}(\bm v)\in\{0,\ldots,B-1\}$ denote the index of $\bm v$, with inverse $\operatorname{unrank}$. For fixed $q$ and $\Delta$, both operations run in $O(\ell)$ time via enumerative coding \cite{Cover1973}.
Write $\operatorname{Rep}_{q,m}(s)$ for the length-$m$ base-$q$ representation of $s$. Define the pointer
\[
\operatorname{Ptr}(i,\bm v)=
\operatorname{Rep}_{q,\ell-\Delta}
\bigl(iB+\operatorname{rank}(\bm v)\bigr)\circ1^\Delta.
\]
Condition~\eqref{eq:WWL-pointer-capacity} ensures that this map is well defined and injective. Every pointer has weight at least $\Delta$.

\begin{algorithm}[t]
\caption{Row-cyclic WWL encoder
$\mathrm{ENC}^{\mathrm{WWL}}_{n,\kappa,\Delta}[\ell]$}
\label{alg:row-cyclic-wwl-encoder}
\begin{algorithmic}[1]
\Require $\bm x\in\Sigma_q^{n-1}$, $\kappa\mid n$, $\Delta\leq\ell$, $2\ell\leq\kappa$, and
\eqref{eq:WWL-pointer-capacity}
\Ensure A row-cyclic $(\ell,\Delta)$-WWL word $\bm y\in\Sigma_q^{(n/\kappa) \times \kappa}$
\State $\bm y\gets\bm x\circ0$, $i\gets0$
\While{$i<n$}
    \State Find the smallest $i'\in\{i,\ldots,n-1\}$ with
    $\wt(\operatorname{rwin}^{(\ell)}_{i'}(\bm y))<\Delta$
    \If{no such $i'$ exists}
        \State \Return $\bm y$
    \EndIf
    \State $\bm v\gets\operatorname{rwin}^{(\ell)}_{i'}(\bm y)$
    \State Write $i'=r\kappa+c$
    \If{$c\leq\kappa-\ell$}
        \State Delete $\bm y[i',i'+\ell-1]$
        \State $i\gets\max\{r\kappa,i'-\ell+1\}$
    \Else
        \State $b\gets\ell-(\kappa-c)$
        \State Delete $\bm y[i',r\kappa+\kappa-1]$ and
        $\bm y[r\kappa,r\kappa+b-1]$ simultaneously
        \State $i\gets r\kappa+\kappa-2\ell+1$
    \EndIf
    \State Append $\operatorname{Ptr}(i',\bm v)$ to $\bm y$
\EndWhile
\State \Return $\bm y$ and view it as an array of size $(n/\kappa)\times \kappa$
\end{algorithmic}
\end{algorithm}

\begin{algorithm}[t]
\caption{Row-cyclic WWL decoder
$\mathrm{DEC}^{\mathrm{WWL}}_{n,\kappa,\Delta}[\ell]$}
\label{alg:row-cyclic-wwl-decoder}
\begin{algorithmic}[1]
\Require $\bm y$ in the image of
Algorithm~\ref{alg:row-cyclic-wwl-encoder}, with the same parameters
\Ensure The original word $\bm x\in\Sigma_q^{n-1}$
\While{the last $\Delta$ symbols of $\bm y$ equal $1^\Delta$}
    \State Read and delete the last $\ell$ symbols as $\bm p$
    \State Interpret the first $\ell-\Delta$ symbols of $\bm p$
    as a base-$q$ integer $s$
    \State $i'\gets\lfloor s/B\rfloor$, $\alpha\gets s\bmod B$
    \State $\bm v\gets\operatorname{unrank}(\alpha)$
    \State Write $i'=r\kappa+c$
    \If{$c\leq\kappa-\ell$}
        \State Insert $\bm v$ at position $i'$
    \Else
        \State $a\gets\kappa-c$
        \State Insert $(v_a,\ldots,v_{\ell-1})$ at position $r\kappa$
        \State Insert $(v_0,\ldots,v_{a-1})$ at position $i'$
    \EndIf
\EndWhile
\State Delete the final zero
\State \Return the remaining word
\end{algorithmic}
\end{algorithm}

Lemma~\ref{lem:row-cyclic-wwl} follows directly from the following conclusion.
\begin{lemma}
Under the stated conditions, Algorithm~\ref{alg:row-cyclic-wwl-encoder}
is an injective map from $\Sigma_q^{n-1}$ to the row-cyclic
$(\ell,\Delta)$-WWL words in $\Sigma_q^n$.
Algorithm~\ref{alg:row-cyclic-wwl-decoder} is its inverse on its image.
For fixed $q,\Delta$, encoding takes $O(n\ell^2)$ word operations
and decoding takes $O(n\ell)$ word operations in the implementation
specified below.
\end{lemma}

\begin{proof}
We first justify the update of the scanning index. At the beginning of
each iteration, no forbidden row-cyclic window starts before $i$.
Suppose that the first forbidden window found starts at $i'=r\kappa+c$, rows preceding row $r$ remain unchanged after the replacement.
\begin{itemize}
    \item If $c\le\kappa-\ell$, a window in row $r$ starting before $\max\{0,c-\ell+1\}$ ends before the deleted block. It is therefore unaffected by the deletion and is not forbidden. Thus the scan can restart from $\max\{r\kappa,i'-\ell+1\}$.
    \item If $c>\kappa-\ell$, set $b=\ell-(\kappa-c)$. After deleting the sufifx and prefix, the retained middle of row $r$ has length $\kappa-\ell$. For every $0\leq s<\kappa-2\ell+1$, in row $r$, the new length-$\ell$ window starting at $s$ equals the old window starting at $s+b\leq \kappa-2\ell+b=c-\ell$. By the minimum choice of $i'$, such a window is not forbidden. Thus the scan can restart from $r\kappa+\kappa-2\ell+1$.
\end{itemize}
Each replacement removes a window of weight at most $\Delta-1$ and appends a pointer of weight at least $\Delta$. Hence the total Hamming weight increases by at least one. Since the word length remains $n$, at most $n$ replacements can occur. Therefore the encoder terminates, and its stopping condition guarantees that no forbidden row-cyclic window remains.

We next verify the decoder. The initial word $\bm{x}\circ 0$ ends in $0$, whereas the word after each replacement ends in $1$. Thus the final $\ell$ symbols identify the pointer from the most recent replacement. For a noncrossing window, the decoder reinserts the recorded word at $i'$. For a crossing window, it first restores its last $b$ symbols at $r\kappa$ and then restores its first $a$ symbols at $i'$. Hence each decoding step reverses the most recent replacement. Repeating this procedure recovers $\bm{x}0$, after which the final zero is removed.

Finally, let $R\leq n$ denote the number of replacements. The restart is at most $2\ell-2$ positions before the selected starting position, the total scan length is $O(n+R\ell)=O(n\ell)$. Since checking each window takes $O(\ell)$ time, the total scanning time is $O(n\ell^2)$. For each replacement, encoding or decoding the position $i<n$ and the forbidden window of length $\ell$ requires $O(\ell)$ time. All replacements thus take $O(R\ell)=O(n\ell)$ time. The total encoding and decoding time is therefore $O(n\ell^2)$. This completes the proof.
\end{proof}

\bibliographystyle{IEEEtran}
\bibliography{references}
\end{document}